\documentclass[letterpaper,11pt]{article}
\usepackage{amsmath, amsthm, amssymb}
\usepackage[usenames, dvipsnames]{color}
\usepackage[normalem]{ulem} 
\usepackage{fullpage}
\usepackage[numbers, sort]{natbib}
\usepackage[noend]{algorithmic}
\usepackage{tikz, pgfplots}
\usepackage{enumerate}
\usepackage{hyperref}
\usepackage{xspace,color}
\usepackage{graphicx}
\usepackage{caption}
\usepackage{subcaption}
\usepackage{enumitem,linegoal}
\usepackage[linesnumbered,ruled,vlined]{algorithm2e}
\usepackage{comment}	
\usepackage{ctable}

\usepackage{mathtools}
\usepackage[flushleft]{threeparttable}
\usepackage{verbatim}
\usepackage{subcaption}
\usepackage{thmtools}
\usepackage{thm-restate}
\usepackage{cleveref}
\usepackage{fullpage}
\usepackage{graphicx}
\usepackage{comment}
\usepackage{mathtools}
\usepackage[linesnumbered,ruled,vlined]{algorithm2e}
\usepackage[utf8]{inputenc}
\usepackage{tikz}
\usepackage{amsthm}
\usepackage{subcaption}
\usepackage{amsfonts}
\usepackage{hyperref}
\usepackage[framemethod=TikZ]{mdframed}
\usepackage{amsthm}
\usepackage{thmtools}
\usepackage{thm-restate}
\usepackage[normalem]{ulem}
\usepackage{filecontents}

\usepackage{bbm}
\usepackage{thm-restate}

\usepackage{footnote}
\makesavenoteenv{tabular}
\makesavenoteenv{table}

\usepackage[margin=1in]{geometry}

\definecolor{crimsonglory}{rgb}{0,0,0}

 \newtheorem{theorem}{Theorem}[section]
 
 \newtheorem{lemma}[theorem]{Lemma}
 \newtheorem{proposition}[theorem]{Proposition}

\makeatletter
\def\GrabProofArgument[#1]{ #1: \egroup\ignorespaces}
\def\proof{\noindent\textbf\bgroup Proof%
	\@ifnextchar[{\GrabProofArgument}{. \egroup\ignorespaces}}

\makeatother

\newenvironment{theorem-repeat}[1]{\begin{trivlist}
		\item[\hspace{\labelsep}{\bf\noindent Theorem \ref{#1} (repeated).}]\em }%
	{\end{trivlist}}

\usetikzlibrary{arrows,shapes,snakes,automata,backgrounds,petri,calc}

\newcommand{\Saeed}[1]{{\color{blue}[Saeed: \textsf{#1}]}}
\newcommand{\Francois}[1]{{\color{red}[Francois: \textsf{#1}]}}

\newcommand{\ket}[1]{|#1\rangle}

\newcommand{\Tt}{\mathcal{T}}
\newcommand{\ceil}[1]{\left\lceil #1 \right\rceil}

\providecommand{\Aa}{\mathcal{A}}

\providecommand{\Cc}{\mathsf{C}}

\newcommand{\poly}{\mathrm{poly}}

\newcommand{\edit}{\mathrm{ud}}

\providecommand{\Cc}{\mathcal{C}}
\providecommand{\Estimator}[1]{\textsf{UlamIndic}(#1)}

\newcounter{proccnt}

\newcommand{\konote}[1]{}

\usepackage[normalem]{ulem} 

\title{Quantum Meets Fine-grained Complexity:\\Sublinear Time Quantum Algorithms for String Problems}

\author{
	Fran{\c c}ois Le Gall\\Nagoya University
	\and Saeed Seddighin\thanks{Supported in part by an Adobe Research Award and a Google Research Gift.}\\
      Toyota Technological Institute at Chicago
}

\begin{document}
	\newcommand{\ignore}[1]{}
\renewcommand{\theenumi}{(\roman{enumi}).}
\renewcommand{\labelenumi}{\theenumi}
\sloppy
\date{}

\maketitle

\thispagestyle{empty}
\begin{abstract}
\setcounter{page}{0}
Longest common substring (\textsf{LCS}), longest palindrome substring (\textsf{LPS}), and Ulam distance (\textsf{UL}) are three fundamental string problems that can be classically solved in near linear time. In this work, we present sublinear time quantum algorithms for these problems along with quantum lower bounds. Our results shed light on a very surprising fact: Although the classic solutions for \textsf{LCS} and \textsf{LPS} are almost identical (via suffix trees), their quantum computational complexities are different. While we give an exact\footnote{Throughout the paper, when we say an exact solution, we mean a solution that does not lose anything in the approximation. However, it may be possible that our algorithm succeeds with probability less than 1 in which case we explicitly mention that.} $\tilde O(\sqrt{n})$ time algorithm for \textsf{LPS}, we prove that \textsf{LCS} needs at least time $\tilde \Omega(n^{2/3})$ even for 0/1 strings.

\thispagestyle{empty}
\newpage
\end{abstract}
\section{Introduction}
Perhaps the earliest questions that were studied in computer science are the algorithmic aspects of string problems. \textit{The edit distance}, \textit{longest common substring}, and \textit{longest palindrome substring} are some of the more famous problems in this category. Efforts to solve these problems  led to the discovery of several fundamental techniques such as \textit{dynamic programming}, \textit{hashing algorithms}, and \textit{suffix trees}. These algorithms have numerous applications in several fields including DNA-sequencing, social media, compiler design, anti-virus softwares, etc.

All of the above problems have received significant attention in the classical setting (see, e.g., ~\cite{masek1980faster,bar2004approximating,batu2006oblivious,andoni2010polylogarithmic,boroujeni2018approximating,chakraborty2018approximating,farach1997optimal,amir2018longest,kociumaka2014sublinear}). Efficient classical algorithms for these problems emerged early on in the 1960s~\cite{cormen2009introduction}. Moreover, thanks to a series of recent developments in fine-grained complexity~\cite{backurs2015edit,abboud2015tight}, we now seem to have a clear understanding of the  classical lower bounds as well. Unless plausible conjectures such as \textsf{SETH}\footnote{The \textit{strong exponential time hypothesis} states that no algorithm can solve the satisfiability problem in time $2^{n(1-\epsilon)}$.} are broken, we do not hope for a substantially better algorithm for edit distance. Failure to extend the fine-grained lower bounds to the quantum setting has left some very interesting open questions behind both in terms of quantum complexity and quantum lower bounds. 

Despite a plethora of new quantum algorithms for various problems (e.g., \cite{ambainis2007quantum,buhrman2006quantum,GroverSTOC96,HHL06,magniez2007quantum,Shor97}), not much attention is given to string problems. Until recently, the only non-trivial quantum algorithm for such problems  was the $\tilde O(\sqrt{n}+\sqrt{m})$ time algorithm of Ramesh and Vinay~\cite{ramesh2003string} for pattern matching where $n$ and $m$ are the sizes of the text and the pattern (we also mention the works~\cite{Ambainis+14,Cleve+12} that consider string problems with nonstandard queries in the quantum setting). Recently, Boroujeni, Ehsani, Ghodsi, Hajiaghayi, and Seddighin~\cite{boroujeni2018approximating} made a clever use of the Grover's search algorithm~\cite{kaye2007introduction} to obtain a constant factor approximation quantum algorithm for edit distance in truly subquadratic time. Shortly after, it was shown by Chakraborty, Das, Goldenberg, Koucky, and Saks~\cite{chakraborty2018approximating}
that a similar technique can be used to obtain a classical solution with the same approximation factor. Several improved classical algorithms have been given for edit distance in recent years~\cite{andoni2020edit,koucky2020constant,brakensiek2020constant} though it is still an open question if a non-trivial quantum algorithm can go beyond what we can do classically for edit distance. 

In this work, we give novel sublinear time quantum algorithms and quantum lower bounds for \textsf{LCS}, \textsf{LPS}, and a special case of edit distance namely \textit{Ulam distance} (\textsf{UL}). All these problems require $\tilde \Omega(n)$ time in the classical setting even if approximate solutions are desired. \textsf{LCS} and \textsf{LPS} can  be solved in linear time via suffix trees~\cite{cormen2009introduction} and there is an $O(n \log n)$ time algorithm for Ulam distance~\cite{cormen2009introduction}. Our results shed light on a very surprising fact: Although the classical solutions for \textsf{LCS} and \textsf{LPS} are almost identical, their quantum computational complexities are different. While we give an exact $\tilde O(\sqrt{n})$ time quantum algorithm for \textsf{LPS}, we prove that any quantum algorithm for \textsf{LCS} needs at least time $\tilde \Omega(n^{2/3})$ even for 0/1 strings. We accompany this with several sublinear time quantum algorithm for \textsf{LCS}. A summary of our results is given in Tables~\ref{table:results} and~\ref{table:results2}. 

\renewcommand{\arraystretch}{1.2}
\begin{table}[!htbp]
	\centering
	\begin{tabular}{|l|c|c|c|}
		\hline
		problem & solution & runtime & reference\\
		\hline
		\textsf{longest common substring} & \textsf{exact} & $\tilde{O}(n^{5/6})$ &  Theorem \ref{theorem:lcs-exact}\\
		\hline
		\textsf{longest palindrome substring} & \textsf{exact} & $\tilde{O}(\sqrt{n})$ & Theorem \ref{theorem:lps}\\
	   \hline
		\textsf{longest common substring} & $1-\epsilon$ \textsf{approximation}  & $\tilde{O}(n^{3/4})$  & Theorem \ref{theorem:lcs-approx} \\		
	   \hline
		\textsf{longest common substring for non-repetitive strings} & \textsf{exact}  & $\tilde{O}(n^{3/4})$  & Theorem \ref{theorem:lcs-exact-perm} \\
		\hline
		\textsf{longest common substring for non-repetitive strings} & $1-\epsilon$ \textsf{approximation}  & $\tilde{O}(n^{2/3})$  & Theorem \ref{theorem:lcs-approx-perm} \\
		\hline
		\textsf{Ulam distance} & $1+\epsilon$ \textsf{approximation}  & $\tilde O(\sqrt{n})$ & Theorem \ref{theorem:ulam} \\
		\hline
	\end{tabular}
	\caption{Our algorithms are shown in this table.}\label{table:results}
\end{table}
\renewcommand{\arraystretch}{1.2}
\begin{table}[!htbp]
	\centering
	\begin{tabular}{|l|c|c|c|}
		\hline
		problem & approximation factor & lower bound & reference\\
		\hline
		\textsf{longest common substring} & $1-\epsilon$ \textsf{approximation} & $\tilde \Omega(n^{2/3})$ & Theorems \ref{theorem:lcs-approx-LB-easy} and \ref{theorem:lcs-approx-LB}\\
		\hline
		\textsf{longest palindrome substring} & $1-\epsilon$ \textsf{approximation} & $\tilde \Omega(\sqrt{n})$ & Theorem \ref{theorem:lps-LB} \\
		\hline
		\textsf{Ulam distance} & $1+\epsilon$ \textsf{approximation} & $\tilde{\Omega}(\sqrt{n})$ & Theorem \ref{theorem:ulam-approx-LB} \\
		\hline
	\end{tabular}
	\caption{This table includes the quantum lower bounds for the problems we study in this work.}\label{table:results2}
\end{table}

\subsection{Related work}
Our work is very similar in spirit to for instance the work of Ambainis, Balodis, Iraids, Kokainis, Prusis, and Vihrovs~\cite{ambainis2019quantum}, where Grover's algorithm is cleverly combined with classical techniques to design more efficient quantum algorithms for dynamic programming problems. This is particularly similar to our approach since we obtain our main results by combining known quantum algorithms with new classical ideas. In the present work, however, we go beyond Grover's algorithm and make use of several other quantum techniques such as element distinctness, pattern matching, amplitude amplification, and amplitude estimation to obtain our improvements. In addition to this, we also develop quantum walks that improve our more general results for some special cases. In particular, our quantum walk for obtaining a $1-\epsilon$ approximate solution for \textsf{LCS} is tight up to logarithmic factors due to a lower-bound we give in Section~\ref{sec:lb}.

Another line of research which is closely related to our work is the study of quantum lower bounds for edit distance~\cite{aviadsblogpost,Buhrman+STACS21}. While a \textsf{SETH}-based quadratic lower bound is known for the classical computation of edit distance, quantum lower bounds are not very strong. Recently, a quantum lower bound of $\Omega(n^{1.5})$ was given by Buhrman, Patro, and Speelman~\cite{Buhrman+STACS21} under a mild assumption. Still no quantum algorithm better than the state of art classical solution (which runs in time $O(n^2/\log ^2n)$~\cite{masek1980faster}) is known for edit distance. The reader can find more details in~\cite{aviadsblogpost}.

While not directly related to string algorithms, another investigation of \textsf{SETH} in the quantum setting is the recent work by Aaronson, Chia, Lin, Wang, and Zhang~\cite{Aaronson+CCC20} that focuses on the quantum complexity of the closest pair problem, a fundamental problem in computational geometry. Interestingly, the upper bounds obtained in that paper also use an approach based on  element distinctness and quantum walks. Despite this, the high-level ideas of our work are substantially different from~\cite{Aaronson+CCC20}  as we utilize several novel properties of \textsf{LCS} and \textsf{LPS} to design our algorithms.

In the classical sequential setting, \textsf{LCS} and \textsf{LPS} can be solved in linear time~\cite{cormen2009introduction}. The solutions are almost identical: we first construct suffix trees for the input strings and then find \textit{the lowest common ancestors} for the tree nodes. Ulam distance can also be solved exactly in time $O(n \log n)$~\cite{cormen2009introduction} which is the best we can hope for via a comparison-based algorithm~\cite{fredman1975computing} or algebraic decision trees~\cite{ramanan1997tight}. Approximation algorithms running in time $\tilde O(n/d+\sqrt{d})$ where $d$ denotes the Ulam distance of the two strings have also been developed \cite{Andoni+SODA10,Naumovitz+SODA17}.

\subsection{Preliminaries}
\paragraph{Description of \textsf{LCS}, \textsf{LPS} and \textsf{UD}.}
In the longest common substring problem (\textsf{LCS}), the input consists of two strings and our goal is to find the longest substring\footnote{In a substring, the characters are next to each other. In other words, the positions of the characters of a substring should make an interval. This is in contrast to subsequence where the positions can be arbitrary.} which is shared between the two strings. We denote the two input strings by $A$ and $B$. We assume that $A$ and $B$ have the same length, which we denote by $n$. We use~$\Sigma$ to denote the alphabet of the strings. For any $\epsilon\in [0,1)$, we say that an algorithm outputs a $(1-\epsilon)$-approximation of the longest common substring if for any input strings $A$ and $B$, it outputs a common substring of length at least $(1-\epsilon)d$, where $d$ is the length of the longest common substring of $A$ and $B$. 

In the longest palindrome substring problem (\textsf{LPS}), the goal is to find the longest substring of a given string $A$ which reads the same both forward and backward. The length of $A$ is also denoted by $n$ and its alphabet by $\Sigma$. For any $\epsilon\in [0,1)$, we say that an algorithm outputs a $(1-\epsilon)$-approximation of the longest palindrome substring if for any input string $A$, it outputs a palindrome substring of length at least $(1-\epsilon)d$, where $d$ is the length of the longest palindrome substring of $A$.

We say that a string of length $n$ over an alphabet $\Sigma$ is \textit{non-repetitive} if no character appears twice in the string (note that this can happen only if $|\Sigma|\ge n$). The Ulam distance is a special case of the edit distance in which the input strings are non-repetitive. Let us now define the problem more formally. In Ulam Distance (\textsf{UD}) we are given two non-repetitive strings $A$ and~$B$ of length $n$, and consider how to transform one of them to the other one. For this purpose we allow two basic operations \textit{character addition} and \textit{character deletion}, each at a unit cost and our goal is to minimize the total cost of the transformation\footnote{Another popular version allows character substitution as a third operation. These two definitions of the Ulam distance only differ by a factor at most 2.}.   We denote by $\edit(A,B)$ the minimum number of such operations needed to transform $A$ into $B$. The goal is to compute $\edit(A,B)$, either exactly or approximately. For any $\epsilon\in [0,1]$, we say that an algorithm outputs a $(1+\epsilon)$-approximation of $\edit(A,B)$ if it outputs some value $r$ such that the inequality $(1-\epsilon)\cdot \edit(A,B)\le r\le (1+\epsilon)\cdot\edit(A,B)$ holds.

\paragraph{General definitions and conventions.}
Throughout the paper, we use  notations $\tilde O(\cdot)$ and $\tilde \Omega(\cdot)$ that hide the polylogarithmic factors in terms of $n$. We always assume that the size of $\Sigma$ is polynomial in $n$ and that each character is encoded using $O(\log n)$ bits. The size of $\Sigma$ thus never appears explicitly in the complexity of our algorithms. We say that a randomized or a quantum algorithm solves a problem like \textsf{LCS}, \textsf{LPS} or \textsf{UL} with high probability if it solves the problem with probability at least 9/10 (this success probability can be easily amplified to $1-1/\poly(n)$ with a logarithmic overhead in the complexity).  

For convenience, we often only compute/approximate the size of the solution as opposed to explicitly giving the solution. However, it is not hard to see that for \textsf{LCS} and \textsf{LPS}, the same algorithms can also give an explicit solution with a logarithmic overhead in the runtime. 
(a solution can be specified by two integers pointing at the interval of the input.)

For a string $X$, we denote by $X[i,j]$ the substring of $X$ that starts from the $i$-th character and ends at the $j$-th character. 
We say a string $X$ is $q$-periodic if we have $X_i = X_{i+q}$ for all $1 \leq i \leq |X|-q$. Moreover, the periodicity of a string $X$ is equal to the smallest number $q > 0$ such that $X$ is $q$-periodic. We also call a non-repetitive string of length~$n$ over an alphabet of size $n$ a \textit{permutation} (it represents a permutation of the set $\Sigma$).

\paragraph{Quantum access to the inputs.}
In the quantum setting, we suppose that the input strings $A$ and $B$ can be accessed directly by a quantum algorithm.
More precisely, we have an oracle~$O_A$ that, for any $i\in\{1,\ldots,n\}$,
any $a\in\Sigma$, and any  $z\in\{0,1\}^\ast$, performs the unitary mapping 
$
O_A\colon \ket{i}\ket{a}\ket{z}\mapsto\ket{i}\ket{a\oplus A[i]}\ket{z},
$
where $\oplus$ denotes an appropriate binary operation defined on $\Sigma$ (e.g., bit-wise parity on the binary encodings of $a$ and $A[i]$).
Similarly we have an oracle $O_B$ that, for any $i\in\{1,\ldots,n\}$,
any $b\in\Sigma$, and any  $z\in\{0,1\}^\ast$, performs the unitary mapping
$
O_B\colon \ket{i}\ket{b}\ket{z}\mapsto\ket{i}\ket{b\oplus B[i]}\ket{z}.
$
Each call to $O_A$ or $O_B$ can be implemented at unit cost. This description corresponds to quantum random access (``QRAM access'') to the input, which is the standard model to investigate the complexity of sublinear time quantum algorithms. 

\section{Results}
We present sublinear time quantum algorithms along with quantum lower bounds for \textsf{LCS}, \textsf{LPS}, and \textsf{UL}. For the most part, the novelty of our work is to make use of existing quantum algorithms to solve our problems. For this purpose, we introduce new classical techniques that significantly differ from the conventional methods. However for a special case of \textsf{LCS}, we design a novel quantum walk that leads to an improvement over our more general solution. We give a brief explanation of this technique later in the section. For now, we start by stating the quantum tools that we use in our algorithms.

\subsection{Quantum components}

\noindent \textbf{Grover's search (\cite{GroverSTOC96, brassard2002quantum}).} Given a function $f\colon [n] \rightarrow \{0,1\}$, Grover's algorithm can find an element  $x\in [n]$ such that $f(x) = 1$ or verify if $f(i) = 0$ for all $i\in [n]$. This quantum algorithm runs in time $\tilde O(\sqrt{n}\cdot T(n))$ and succeeds with probability $9/10$ (the success probability can be increased to $1-1/\poly(n)$ with only a logarithmic overhead). Here, $T(n)$ represents the time complexity of computing $f(i)$ for one given element $i\in [n]$. Additionally, distinguishing between the case where $f(x)=1$ holds for at least $m$ elements (for some value $1\le m\le n$) and the case where $f(i) = 0$ for all $i\in [n]$ can be done in time $\tilde O(\sqrt{n/m}\cdot T(n))$.
\\[0.1cm]

\noindent \textbf{Pattern matching (\cite{ramesh2003string}).} Let $P$ and $S$ be a pattern and a text of lengths $n$ and $m$ respectively. One can either verify that $P$ does not appear as a substring in $S$ or find the leftmost (rightmost) occurrence of $P$ in $S$ in time $\tilde O(\sqrt{n}+\sqrt{m})$ via a quantum algorithm. The algorithm gives a correct solution with probability at least $9/10$.\\[0.1cm]

\noindent \textbf{Element distinctness (\cite{ambainis2007quantum}).} Let $X$ and $Y$ be two lists of size $n$ and $f\colon (X \cup Y) \rightarrow \mathbb{N}$ be a function that is used to compare the elements of $X$ and $Y$\footnote{In the standard definition of element distinctness, we are given a single list of elements and the goal is to find out if two elements in the list are equal. The present definition, also known as claw finding, is slightly more general --- for completeness we discuss how the upper bound $\tilde O(n^{2/3} \cdot T(n))$ is obtained for our version as well in Section \ref{sub:walks}.}.
 There is a quantum algorithm that finds (if any) an $(x,y)$ pair such that $x \in X$, $y \in Y$ and $f(x) = f(y)$. The algorithm succeeds with probability at least 9/10 and has running time $\tilde O(n^{2/3} \cdot T(n))$, where $T(n)$ represents the time needed to answer to the following question: Given $\alpha,\beta \in X \cup Y$ is $f(\alpha) = f(\beta)$ and if not which one is smaller? \\[0.1cm]


\noindent \textbf{Amplitude amplification (\cite{brassard2002quantum}).} Let $\mathsf{Q}$ be a decision problem and $\mathcal{A}$ be a quantum algorithm that solves $\mathsf{Q}$ with one-sided error and success probability $0 < p < 1$ (i.e., on a yes-instance $\mathcal{A}$ always accepts, while on a no-instance $\mathcal{A}$ rejects with probability $p$). Let $T$ be the runtime of $\mathcal{A}$. One can design a quantum algorithm for $Q$ with runtime $O(T /\sqrt{p})$ that solves $\mathsf{Q}$ with one-sided error and success probability at least $9/10$.\\[0.1cm]



\noindent \textbf{Amplitude estimation (\cite{brassard2002quantum}).} Let  $\mathcal{A}$ be a quantum algorithm that outputs 1 with probability $0 < p < 1$ and returns $0$ with probability $1-p$. Let $T$ be the time needed for $\mathcal{A}$ to generate its output. For any $\alpha>0$, one can design a quantum algorithm with runtime $O(T /(\alpha\sqrt{p}))$ that outputs with probability at least 9/10 an estimate $\tilde p$ such that $(1-\alpha)p\le \tilde p\le (1+\alpha) p.$

\subsection{\textsf{LCS} and \textsf{LPS}}
In this section, we outline the ideas for obtaining sublinear time algorithms for \textsf{LCS} and \textsf{LPS}. We begin as a warm up by giving a simple exact algorithm for \textsf{LCS} that runs in sublinear time when the solution size is small. Next, we explain our techniques for the cases that the solution size is large (at a high-level, this part of the algorithm is very similar in both \textsf{LCS} and \textsf{LPS}). In our algorithms, we do a binary search on the size of the solution. We denote this value by $d$. Thus, by losing an $O(\log n)$ factor in the runtime, we reduce the problem to verifying if a solution of size at least $d$ exists for our problem instance.

\paragraph*{Exact quantum algorithm for \textsf{LCS} (small $d$).}
For small $d$, we use element distinctness to solve \textsf{LCS}. 
Let $|\Sigma|$ be the size of the alphabet and $v:\Sigma \rightarrow [0,|\Sigma|-1]$ be a function that maps every element of the alphabet to a distinct number in range $0 \ldots |\Sigma|-1$. In other words, $v$ is a perfect hashing for the characters. We extend this definition to substrings of the two strings so that two substrings $t$ and $t’$ are equal if and only if we have $v(t) = v(t’)$. From the two strings, we then make two sets of numbers $S_A$ and $S_B$ each having $n-d+1$ elements. Element $i$ of set $S_A$ is a pair $(A,i)$ whose value is equal to $v(A[i,i+d-1])$ and similarly element $i$ of $S_B$ is a pair $(B,i)$ whose value is equal to $v(B[i,i+d-1])$. 
The two sets contain elements with equal values if and only if the two strings have a common substring of size~$d$. Therefore, by solving element distinctness for $S_A$ and $S_B$ we can find out if the size of the solution is at least~$d$. Although element distinctness can be solved in time $\tilde O(n^{2/3})$ when element comparison can be implemented in $\tilde O(1)$ time, our algorithm needs more runtime since comparing elements takes time $\omega(1)$. More precisely, each comparison can be implemented in time $\tilde O(\sqrt{d})$ in the following way: In order to compare two elements of the two sets, we can use Grover's search to find out if the two substrings are different in any position and if so we can find the leftmost position in time $\tilde O(\sqrt{d})$. Thus, it takes time $\tilde O(\sqrt{d})$ to compare two elements, which results in runtime $\tilde O(n^{2/3} \sqrt{d})$.

\paragraph{$1-\epsilon$ approximation for \textsf{LCS} and \textsf{LPS} (large $d$).} We use another technique to solve \textsf{LCS} and \textsf{LPS} when the solution is large. While for \textsf{LPS}, this new idea alone gives an optimal solution, for \textsf{LCS} we need to combine it with the previous algorithm to make sure the runtime is sublinear. Let us focus on \textsf{LCS} first. For a constant $0 < \epsilon < 1$, we define $d' = (1-\epsilon)d$ and randomly draw a substring of length $d'$ from $A$. We denote this substring by $P$. More precisely, we sample an $1 \leq i \leq n-d'+1$ uniformly at random and set $P = A[i,i+d'-1]$. Assuming the solution size is at least $d$, it follows that $P$ is part of a solution with probability at least $\epsilon d/n$. Then, by searching this substring in $B$, we can find a solution of size $d'$.

We use the pattern matching quantum algorithm of Ramesh and Vinay~\cite{ramesh2003string} to search $P$ in $B$. This takes time $\tilde O(\sqrt{n})$ since $|P| \leq |B| = n$. Moreover, the success probability of this algorithm is $\Omega(d/n)$ and therefore by amplitude amplification, we can improve the success rate to $9/10$ by only losing a factor of $O(\sqrt{n/d})$ in the runtime. Thus, if the solution is at least $d$, we can obtain a solution of size at least $(1-\epsilon)d$ in time $\tilde O(\sqrt{n/d}\cdot \sqrt{n})=\tilde O(n/\sqrt{d})$. Notice that the runtime is sublinear when $d$ is large.

The same technique can be used to approximate \textsf{LPS}. Similarly, we define $d' = (1-\epsilon)d$ for some constant $0 < \epsilon < 1$ and draw a random substring of size $d'$ from $A$. With the same argument, provided that the solution size is at least $d$, the probability that $P$ is part of an optimal solution is at least $\Omega(d/n)$. We show in Section \ref{section:lps} that by searching the reverse of $P$ in its neighbourhood we are able to find a solution of at least $d'$. This step of the algorithm slightly differs from \textsf{LCS} in that we only search the reverse of $P$ in the area at most $d$ away from $P$. Thus, both the text and the pattern are of size $O(d)$ and therefore the search can be done in time $\tilde O(\sqrt{d})$. By utilizing amplitude amplification, we can obtain an algorithm with runtime $\tilde O(\sqrt{n/d}\cdot\sqrt{d})=\tilde O(\sqrt{n})$ and approximation factor $1-\epsilon$. 

\paragraph{From $1-\epsilon$ approximation to exact solution.}
We further develop a clever technique to obtain an exact solution with the above ideas. We first focus on \textsf{LCS} to illustrate this new technique. The high-level intuition is the following: After sampling $P$ from $A$ and searching $P$ in $B$, if the pattern appears only once (or a small number of times) in $B$ then by extending the matching parts of $B$ from both ends we may find a common substring of size $d$ (see Figure \ref{fig:5}). 
\begin{figure}[!htbp]
	\centering

\tikzset{every picture/.style={line width=0.75pt}} 

\begin{tikzpicture}[x=0.75pt,y=0.75pt,yscale=-1,xscale=1]

\draw   (52,60) -- (607,60) -- (607,91) -- (52,91) -- cycle ;
\draw  [fill={rgb, 255:red, 255; green, 210; blue, 132 }  ,fill opacity=0.5 ][dash pattern={on 0.84pt off 2.51pt}] (219.5,64) -- (420,64) -- (420,87) -- (219.5,87) -- cycle ;
\draw   (53,151) -- (608,151) -- (608,182) -- (53,182) -- cycle ;
\draw  [fill={rgb, 255:red, 255; green, 210; blue, 132 }  ,fill opacity=0.5 ][dash pattern={on 0.84pt off 2.51pt}] (294.5,155) -- (495,155) -- (495,178) -- (294.5,178) -- cycle ;
\draw  [fill={rgb, 255:red, 208; green, 2; blue, 27 }  ,fill opacity=0.4 ][dash pattern={on 0.84pt off 2.51pt}] (192.5,64) -- (219,64) -- (219,87) -- (192.5,87) -- cycle ;
\draw  [fill={rgb, 255:red, 126; green, 211; blue, 33 }  ,fill opacity=0.4 ][dash pattern={on 0.84pt off 2.51pt}] (419.5,64) -- (511.5,64) -- (511.5,87) -- (419.5,87) -- cycle ;
\draw  [fill={rgb, 255:red, 208; green, 2; blue, 27 }  ,fill opacity=0.4 ][dash pattern={on 0.84pt off 2.51pt}] (267.5,155) -- (294,155) -- (294,178) -- (267.5,178) -- cycle ;
\draw  [fill={rgb, 255:red, 126; green, 211; blue, 33 }  ,fill opacity=0.4 ][dash pattern={on 0.84pt off 2.51pt}] (494.5,155) -- (586.5,155) -- (586.5,178) -- (494.5,178) -- cycle ;
\draw    (197.5,43.99) -- (508.5,43.01) ;
\draw [shift={(510.5,43)}, rotate = 539.8199999999999] [color={rgb, 255:red, 0; green, 0; blue, 0 }  ][line width=0.75]    (10.93,-3.29) .. controls (6.95,-1.4) and (3.31,-0.3) .. (0,0) .. controls (3.31,0.3) and (6.95,1.4) .. (10.93,3.29)   ;
\draw [shift={(195.5,44)}, rotate = 359.82] [color={rgb, 255:red, 0; green, 0; blue, 0 }  ][line width=0.75]    (10.93,-3.29) .. controls (6.95,-1.4) and (3.31,-0.3) .. (0,0) .. controls (3.31,0.3) and (6.95,1.4) .. (10.93,3.29)   ;
\draw    (272.5,198.99) -- (583.5,198.01) ;
\draw [shift={(585.5,198)}, rotate = 539.8199999999999] [color={rgb, 255:red, 0; green, 0; blue, 0 }  ][line width=0.75]    (10.93,-3.29) .. controls (6.95,-1.4) and (3.31,-0.3) .. (0,0) .. controls (3.31,0.3) and (6.95,1.4) .. (10.93,3.29)   ;
\draw [shift={(270.5,199)}, rotate = 359.82] [color={rgb, 255:red, 0; green, 0; blue, 0 }  ][line width=0.75]    (10.93,-3.29) .. controls (6.95,-1.4) and (3.31,-0.3) .. (0,0) .. controls (3.31,0.3) and (6.95,1.4) .. (10.93,3.29)   ;
\draw    (389,169) -- (389,139) ;
\draw [shift={(389,136)}, rotate = 450] [fill={rgb, 255:red, 0; green, 0; blue, 0 }  ][line width=0.08]  [draw opacity=0] (8.93,-4.29) -- (0,0) -- (8.93,4.29) -- cycle    ;

\draw (27,67.4) node [anchor=north west][inner sep=0.75pt]    {$A$};
\draw (28,158.4) node [anchor=north west][inner sep=0.75pt]    {$B$};
\draw (317,67.4) node [anchor=north west][inner sep=0.75pt]    {$P$};
\draw (233,17) node [anchor=north west][inner sep=0.75pt]   [align=left] { fixed common substring of size $\displaystyle d$};
\draw (330,209) node [anchor=north west][inner sep=0.75pt]   [align=left] {fixed common substring of size $\displaystyle d$};
\draw (292,114) node [anchor=north west][inner sep=0.75pt]   [align=left] {the only appearance of $\displaystyle P\ \textrm{in}\ B\ $};

\end{tikzpicture}

\caption{When $P$ appears in $B$ only once.}\label{fig:5}
\end{figure}
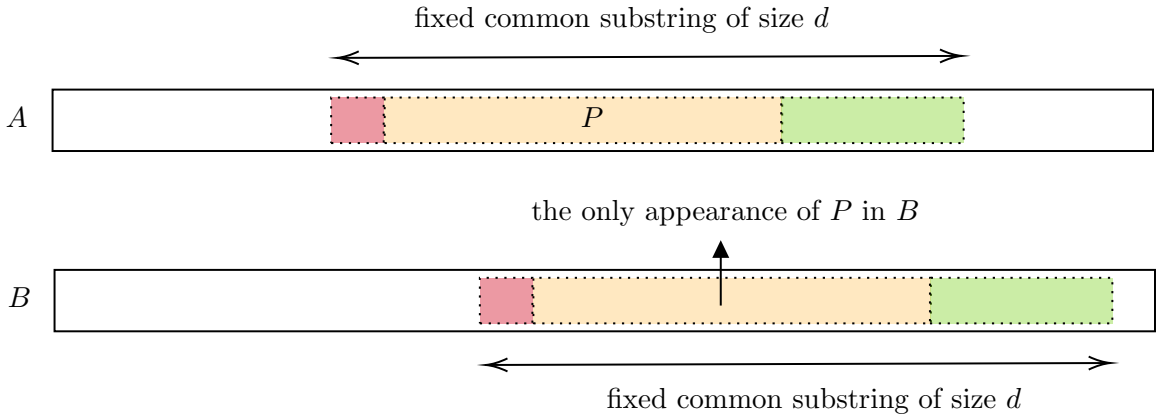
Thus intuitively, the challenging case is when there are several occurrences of $P$ in $B$. The key observation is that since $|P|$ is large and there are several places that $P$ appears in $B$ then they must overlap (see Figure \ref{fig:6}). This gives a very convenient approach to tackle the problem.  Assume that $P$ appears at positions $i$ and $j > i$ of $B$ and these parts are overlapping. This implies that $P$ is certainly $(j-i)$-periodic. To see how this enables us to solve the problem, assume that $P$ is $x$-periodic and that a large continuous area of $B$ is covered by occurrences of $P$. It follows that except for the boundary cases,~$P$ appears as parts of that interval that are exactly $x$ units away. Therefore, by detecting such an interval and computing the periodicity of $P$, one can determine \textbf{all} occurrences of $P$ in the interval almost as quickly as finding one occurrence of $P$.
\begin{figure}[!htbp]
	\centering

\tikzset{every picture/.style={line width=0.75pt}} 

\begin{tikzpicture}[x=0.75pt,y=0.75pt,yscale=-1,xscale=1]

\draw   (52,60) -- (607,60) -- (607,91) -- (52,91) -- cycle ;
\draw  [fill={rgb, 255:red, 255; green, 210; blue, 132 }  ,fill opacity=0.5 ][dash pattern={on 0.84pt off 2.51pt}] (219.5,64) -- (420,64) -- (420,87) -- (219.5,87) -- cycle ;
\draw   (53,131) -- (608,131) -- (608,162) -- (53,162) -- cycle ;
\draw  [fill={rgb, 255:red, 255; green, 210; blue, 132 }  ,fill opacity=0.5 ][dash pattern={on 0.84pt off 2.51pt}] (294.5,135) -- (495,135) -- (495,158) -- (294.5,158) -- cycle ;
\draw  [fill={rgb, 255:red, 208; green, 2; blue, 27 }  ,fill opacity=0.4 ][dash pattern={on 0.84pt off 2.51pt}] (192.5,64) -- (219,64) -- (219,87) -- (192.5,87) -- cycle ;
\draw  [fill={rgb, 255:red, 126; green, 211; blue, 33 }  ,fill opacity=0.4 ][dash pattern={on 0.84pt off 2.51pt}] (419.5,64) -- (511.5,64) -- (511.5,87) -- (419.5,87) -- cycle ;
\draw    (197.5,43.99) -- (508.5,43.01) ;
\draw [shift={(510.5,43)}, rotate = 539.8199999999999] [color={rgb, 255:red, 0; green, 0; blue, 0 }  ][line width=0.75]    (10.93,-3.29) .. controls (6.95,-1.4) and (3.31,-0.3) .. (0,0) .. controls (3.31,0.3) and (6.95,1.4) .. (10.93,3.29)   ;
\draw [shift={(195.5,44)}, rotate = 359.82] [color={rgb, 255:red, 0; green, 0; blue, 0 }  ][line width=0.75]    (10.93,-3.29) .. controls (6.95,-1.4) and (3.31,-0.3) .. (0,0) .. controls (3.31,0.3) and (6.95,1.4) .. (10.93,3.29)   ;
\draw    (189.5,140) -- (339.33,216.63) ;
\draw [shift={(342,218)}, rotate = 207.09] [fill={rgb, 255:red, 0; green, 0; blue, 0 }  ][line width=0.08]  [draw opacity=0] (8.93,-4.29) -- (0,0) -- (8.93,4.29) -- cycle    ;
\draw  [fill={rgb, 255:red, 255; green, 210; blue, 132 }  ,fill opacity=0.5 ][dash pattern={on 0.84pt off 2.51pt}] (178.5,146) -- (379,146) -- (379,169) -- (178.5,169) -- cycle ;
\draw  [fill={rgb, 255:red, 255; green, 210; blue, 132 }  ,fill opacity=0.5 ][dash pattern={on 0.84pt off 2.51pt}] (181.5,135) -- (382,135) -- (382,158) -- (181.5,158) -- cycle ;
\draw  [fill={rgb, 255:red, 255; green, 210; blue, 132 }  ,fill opacity=0.5 ][dash pattern={on 0.84pt off 2.51pt}] (371.5,146) -- (572,146) -- (572,169) -- (371.5,169) -- cycle ;
\draw  [fill={rgb, 255:red, 255; green, 210; blue, 132 }  ,fill opacity=0.5 ][dash pattern={on 0.84pt off 2.51pt}] (168.5,156) -- (369,156) -- (369,179) -- (168.5,179) -- cycle ;
\draw  [fill={rgb, 255:red, 255; green, 210; blue, 132 }  ,fill opacity=0.5 ][dash pattern={on 0.84pt off 2.51pt}] (358.5,156) -- (559,156) -- (559,179) -- (358.5,179) -- cycle ;
\draw    (389,146) -- (348.48,217.39) ;
\draw [shift={(347,220)}, rotate = 299.58] [fill={rgb, 255:red, 0; green, 0; blue, 0 }  ][line width=0.08]  [draw opacity=0] (8.93,-4.29) -- (0,0) -- (8.93,4.29) -- cycle    ;
\draw    (459,169) -- (364.67,217.63) ;
\draw [shift={(362,219)}, rotate = 332.73] [fill={rgb, 255:red, 0; green, 0; blue, 0 }  ][line width=0.08]  [draw opacity=0] (8.93,-4.29) -- (0,0) -- (8.93,4.29) -- cycle    ;
\draw    (275,152) -- (344.82,217.94) ;
\draw [shift={(347,220)}, rotate = 223.36] [fill={rgb, 255:red, 0; green, 0; blue, 0 }  ][line width=0.08]  [draw opacity=0] (8.93,-4.29) -- (0,0) -- (8.93,4.29) -- cycle    ;
\draw    (511,152) -- (355.77,216.84) ;
\draw [shift={(353,218)}, rotate = 337.33000000000004] [fill={rgb, 255:red, 0; green, 0; blue, 0 }  ][line width=0.08]  [draw opacity=0] (8.93,-4.29) -- (0,0) -- (8.93,4.29) -- cycle    ;

\draw (27,67.4) node [anchor=north west][inner sep=0.75pt]    {$A$};
\draw (28,138.4) node [anchor=north west][inner sep=0.75pt]    {$B$};
\draw (317,67.4) node [anchor=north west][inner sep=0.75pt]    {$P$};
\draw (233,17) node [anchor=north west][inner sep=0.75pt]   [align=left] { fixed common substring of size $\displaystyle d$};
\draw (238,223) node [anchor=north west][inner sep=0.75pt]   [align=left] {several appearances of $\displaystyle P\ in\ B\ $};

\end{tikzpicture}

\caption{When $P$ appears several times in $B$.}\label{fig:6}
\end{figure}
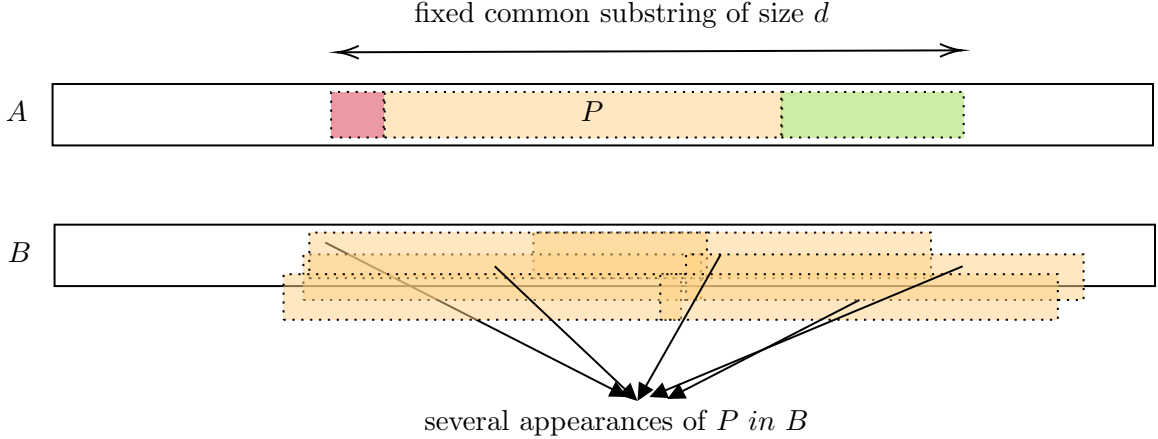

Since the techniques are involved, we defer the details to Section \ref{sec:lcs} and here we just mention some intuition about its complexity analysis: It follows from the above observations that when several occurrences of $P$ cover an entire interval $I$ of $B$, then we only need to consider $O(1)$ places in $I$ that may correspond to an optimal solution. Moreover, the length of every such interval is at least $d$. Thus, there are only $n/d$ places of $B$ that we need to take into account in our algorithm. Therefore at a high-level, when $d$ is large (say~$\Omega(n)$), the only non-negligible cost we pay is for the pattern matching which takes time $\tilde O(\sqrt{n})$. We show in Section \ref{sec:lcs} that as $d$ becomes smaller the runtime increases; more precisely, when $d = n/\alpha$ the runtime increases by a factor of $O(\sqrt{\alpha})$ and thus the overall runtime of the algorithm is $\tilde O(\sqrt{n}\cdot \sqrt{\alpha})=\tilde O(n/\sqrt{d})$.
A very similar argument can be used to shave the $1-\epsilon$ approximation factor from the \textsf{LPS} algorithm as well.

\newcommand{\TheoremLCSExact}
{
The longest common substring of two strings of size $n$ can be computed with high probability by a quantum algorithm in time $\tilde O(n^{5/6})$. 
}

\newcommand{\TheoremLCSApprox}
{
	For any constant $0 < \epsilon < 1$, the longest common substring of two strings of size $n$ can be approximated within a factor $1-\epsilon$ with high probability by a quantum algorithm in time $\tilde O(n^{3/4})$.
}

\newcommand{\TheoremLPS}
{
The longest palindrome substring of a string of size $n$ can be computed with high probability by a quantum algorithm in time $\tilde O(\sqrt{n})$.
}

\newcommand{\TheoremUlam}
{
For any constant $\epsilon>0$, there exists a quantum algorithm that computes with high probability a $(1+\epsilon)$-approximation of the Ulam distance of two non-repetitive strings in time $\tilde O(\sqrt{n})$.
}

\newcommand{\TheoremLCSExactPerm}
{
The longest common substring of two non-repetitive strings of size $n$ can be computed with high probability by a quantum algorithm in time $\tilde O(n^{3/4})$.
}

\newcommand{\TheoremLCSApproxPerm}
{
For any constant $0 < \epsilon < 1$, the longest common substring of two non-repetitive strings of size $n$ can be approximated within a factor $1-\epsilon$ with high probability by a quantum algorithm in time $\tilde O(n^{2/3})$.

}

To summarize, we obtain the following theorems. Theorem~\ref{theorem:lcs-exact} combines the two algorithms we just described: if $d$ is larger than $n^{1/3}$ we use the second algorithm with runtime $\tilde O(n/\sqrt{d})$ otherwise we use the first algorithm with runtime $\tilde O(n^{2/3}\sqrt{d})$.
\begin{theorem}
	\label{theorem:lcs-exact}
	\TheoremLCSExact
\end{theorem}
  
\begin{theorem}
	\label{theorem:lps}
	\TheoremLPS
\end{theorem}



We accompany Theorems~\ref{theorem:lcs-exact} and~\ref{theorem:lps} with quantum lower bounds (all the lower bounds are proven in Section~\ref{sec:lb}). Intuitively, obtaining a solution with time better than $\tilde O(\sqrt{n})$ is impossible for either problem due to a reduction to searching unordered sets. This makes our solution for \textsf{LPS} optimal up to subpolynomial factors. For \textsf{LCS}, an improved lower bound of $\tilde \Omega(n^{2/3})$ can be obtained via a reduction from element distinctness. However, the gap is still open between our upper bound of $\tilde O(n^{5/6})$ and lower bound of $\tilde \Omega(n^{2/3})$. Thus, we aim to improve our upper bound by considering approximate solutions and special cases. In the following, we briefly explain these results. 

\paragraph{Improved $1-\epsilon$ approximation for \textsf{LCS}.}
One way to obtain a better algorithm for \textsf{LCS} is by considering $1-\epsilon$ approximation algorithms. Note that the quantum algorithm for element distinctness is based on quantum walks, and our solution for \textsf{LCS} is obtained via a reduction to element distinctness. The runtime of quantum walks can be improved when multiple solutions are present for a problem. If instead of a solution of size $d$ which is exact, we resort to solutions of size $(1-\epsilon)d$, we are guaranteed to have at least $\epsilon d$ solutions. Thus, intuitively this should help us improve the runtime of our algorithm.

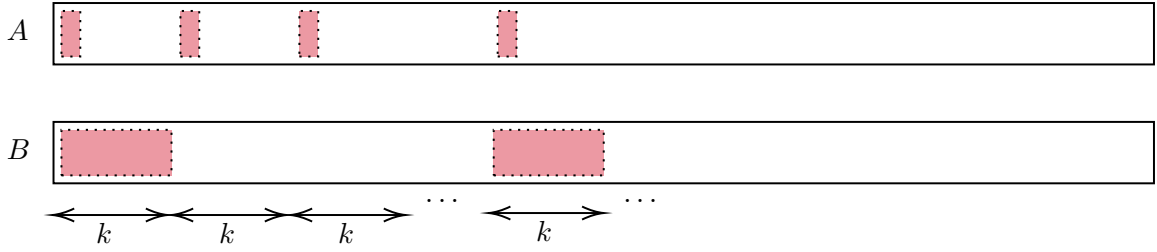
\begin{figure}[!htbp]
	\centering

\tikzset{every picture/.style={line width=0.75pt}} 

\begin{tikzpicture}[x=0.75pt,y=0.75pt,yscale=-1,xscale=1]

\draw   (52,20) -- (607,20) -- (607,51) -- (52,51) -- cycle ;
\draw   (52,80) -- (607,80) -- (607,111) -- (52,111) -- cycle ;
\draw    (53.5,127) -- (107.5,127) ;
\draw [shift={(109.5,127)}, rotate = 180] [color={rgb, 255:red, 0; green, 0; blue, 0 }  ][line width=0.75]    (10.93,-3.29) .. controls (6.95,-1.4) and (3.31,-0.3) .. (0,0) .. controls (3.31,0.3) and (6.95,1.4) .. (10.93,3.29)   ;
\draw [shift={(51.5,127)}, rotate = 0] [color={rgb, 255:red, 0; green, 0; blue, 0 }  ][line width=0.75]    (10.93,-3.29) .. controls (6.95,-1.4) and (3.31,-0.3) .. (0,0) .. controls (3.31,0.3) and (6.95,1.4) .. (10.93,3.29)   ;
\draw    (113.5,127) -- (167.5,127) ;
\draw [shift={(169.5,127)}, rotate = 180] [color={rgb, 255:red, 0; green, 0; blue, 0 }  ][line width=0.75]    (10.93,-3.29) .. controls (6.95,-1.4) and (3.31,-0.3) .. (0,0) .. controls (3.31,0.3) and (6.95,1.4) .. (10.93,3.29)   ;
\draw [shift={(111.5,127)}, rotate = 0] [color={rgb, 255:red, 0; green, 0; blue, 0 }  ][line width=0.75]    (10.93,-3.29) .. controls (6.95,-1.4) and (3.31,-0.3) .. (0,0) .. controls (3.31,0.3) and (6.95,1.4) .. (10.93,3.29)   ;
\draw    (173.5,127) -- (227.5,127) ;
\draw [shift={(229.5,127)}, rotate = 180] [color={rgb, 255:red, 0; green, 0; blue, 0 }  ][line width=0.75]    (10.93,-3.29) .. controls (6.95,-1.4) and (3.31,-0.3) .. (0,0) .. controls (3.31,0.3) and (6.95,1.4) .. (10.93,3.29)   ;
\draw [shift={(171.5,127)}, rotate = 0] [color={rgb, 255:red, 0; green, 0; blue, 0 }  ][line width=0.75]    (10.93,-3.29) .. controls (6.95,-1.4) and (3.31,-0.3) .. (0,0) .. controls (3.31,0.3) and (6.95,1.4) .. (10.93,3.29)   ;
\draw    (273.5,126) -- (327.5,126) ;
\draw [shift={(329.5,126)}, rotate = 180] [color={rgb, 255:red, 0; green, 0; blue, 0 }  ][line width=0.75]    (10.93,-3.29) .. controls (6.95,-1.4) and (3.31,-0.3) .. (0,0) .. controls (3.31,0.3) and (6.95,1.4) .. (10.93,3.29)   ;
\draw [shift={(271.5,126)}, rotate = 0] [color={rgb, 255:red, 0; green, 0; blue, 0 }  ][line width=0.75]    (10.93,-3.29) .. controls (6.95,-1.4) and (3.31,-0.3) .. (0,0) .. controls (3.31,0.3) and (6.95,1.4) .. (10.93,3.29)   ;
\draw  [fill={rgb, 255:red, 208; green, 2; blue, 27 }  ,fill opacity=0.4 ][dash pattern={on 0.84pt off 2.51pt}] (56,24) -- (65.5,24) -- (65.5,47) -- (56,47) -- cycle ;
\draw  [fill={rgb, 255:red, 208; green, 2; blue, 27 }  ,fill opacity=0.4 ][dash pattern={on 0.84pt off 2.51pt}] (116,24) -- (125.5,24) -- (125.5,47) -- (116,47) -- cycle ;
\draw  [fill={rgb, 255:red, 208; green, 2; blue, 27 }  ,fill opacity=0.4 ][dash pattern={on 0.84pt off 2.51pt}] (176,24) -- (185.5,24) -- (185.5,47) -- (176,47) -- cycle ;
\draw  [fill={rgb, 255:red, 208; green, 2; blue, 27 }  ,fill opacity=0.4 ][dash pattern={on 0.84pt off 2.51pt}] (276,24) -- (285.5,24) -- (285.5,47) -- (276,47) -- cycle ;
\draw  [fill={rgb, 255:red, 208; green, 2; blue, 27 }  ,fill opacity=0.4 ][dash pattern={on 0.84pt off 2.51pt}] (56,84) -- (111.5,84) -- (111.5,107) -- (56,107) -- cycle ;
\draw  [fill={rgb, 255:red, 208; green, 2; blue, 27 }  ,fill opacity=0.4 ][dash pattern={on 0.84pt off 2.51pt}] (274,84) -- (329.5,84) -- (329.5,107) -- (274,107) -- cycle ;

\draw (27,27.4) node [anchor=north west][inner sep=0.75pt]    {$A$};
\draw (27,87.4) node [anchor=north west][inner sep=0.75pt]    {$B$};
\draw (72,129.4) node [anchor=north west][inner sep=0.75pt]    {$k$};
\draw (134,129.4) node [anchor=north west][inner sep=0.75pt]    {$k$};
\draw (194,129.4) node [anchor=north west][inner sep=0.75pt]    {$k$};
\draw (238,117.4) node [anchor=north west][inner sep=0.75pt]    {$\dotsc $};
\draw (294,128.4) node [anchor=north west][inner sep=0.75pt]    {$k$};
\draw (338,117.4) node [anchor=north west][inner sep=0.75pt]    {$\dotsc $};

\end{tikzpicture}

	\caption{$k = \epsilon \sqrt{d}$. Only the elements colored in red are included in the two sets.}\label{fig:f4-intro}
\end{figure}

Although the intuition comes from the inner workings of the quantum techniques, we are actually able to improve the runtime with a completely combinatorial idea. For small $d$, instead of constructing two sets $S_A$ and $S_B$ with $n-d+1$ elements, we construct two sets of size $O(n/\sqrt{d})$ and prove that if the two sets have an element in common, then there is a solution of size at least $(1-\epsilon)d$. The construction of the two sets is exactly the same as the construction of $S_A$ and $S_B$ except that only some of the elements are present in the new subsets (see Figure \ref{fig:f4-intro} for an illustration of the construction).

We prove that since each set now has size $O(n/\sqrt{d})$, if the two strings have an \textsf{LCS} of size $d$ then a solution of size $(1-\epsilon)d$ can be found in time $\tilde O((n/\sqrt{d})^{2/3}\sqrt{d})=\tilde O(n^{2/3}d^{1/6})$. This combined with the $\tilde O(n/\sqrt{d})$ algorithm for large $d$ gives us a $1-\epsilon$ approximation algorithm with runtime $\tilde O(n^{3/4})$.

\begin{theorem}
	\label{theorem:lcs-approx}
	\TheoremLCSApprox
\end{theorem}

\paragraph{\textsf{LCS} for non-repetitive strings.}
For all string problems, a special case which is of particular interest is when the characters are different. For instance, although  DNA's consist of only 4 characters, one can make the representation more informative by assigning a symbol to every meaningful block of the sequence. This way, there is only little chance a character appears several times in the sequence. Similarly, in text recognition, one may rather represent every word with a symbol of the alphabet resulting in a huge alphabet and strings with low repetitions. These scenarios are motivating examples for the study of edit distance and longest common subsequence under this assumption, known respectively as Ulam distance \cite{Andoni+SODA10,Charikar06,Naumovitz+SODA17} and longest increasing subsequence (which has been the target of significant research by the string algorithms community --- see, e.g.,~\cite{Saeed+STOC20} for references).

We thus consider \textsf{LCS} for input strings $A$ and $B$ that are non-repetitive (an important special case is when $A$ and $B$ are permutations). We show that there exists an $\tilde O(n^{3/4})$-time quantum algorithm for exact $\textsf{LCS}$ and an $\tilde O(n^{2/3})$-time quantum algorithm  for approximate $\textsf{LCS}$. This significantly improves the generic results of Theorems \ref{theorem:lcs-exact} and \ref{theorem:lcs-approx} and, for approximate $\textsf{LCS}$, this matches (up to possible polylogarithmic factors) the lower bound of Theorem~\ref{theorem:lcs-approx-LB-easy}. 
\begin{theorem}\label{theorem:lcs-exact-perm}
\TheoremLCSExactPerm
\end{theorem}
\begin{theorem}\label{theorem:lcs-approx-perm}
\TheoremLCSApproxPerm
\end{theorem}

The improvements for non-repetitive strings are obtained by improving the complexity of the first part of the algorithm used for general $\mathsf{LCS}$ from $\tilde O(n^{2/3}\sqrt{d})$ to $\tilde O(n^{2/3}+\sqrt{nd})$ for the case of exact $\textsf{LCS}$, and from $\tilde O(n^{2/3}d^{1/6})$ to $\tilde O(n^{2/3}/d^{1/3}+\sqrt{n})$ for approximate $\textsf{LCS}$. We now briefly describe how we achieve such improvements. 

Let us first consider exact $\textsf{LCS}$. The $\tilde O(n^{2/3}\sqrt{d})$-time quantum algorithm described above was based on a ``black-box reduction'' to element distinctness, i.e., we used the quantum algorithm for element distinctness (which is based, as already mentioned, on a technique known as quantum walk) as a black-box. In comparison, our new algorithm is constructed by designing a quantum walk especially tailored for our problem. More precisely, we use the approach by Magniez, Nayak, Roland and Santha \cite{Magniez+SICOMP11} to design a quantum walk over the Johnson graph (more precisely, we work with a graph defined as the direct product of two Johnson graphs, which is more convenient for our purpose). 


We say that a pair $(i,j)\in [n-d+1]\times [n-d+1]$ is marked if there is a common substring of length $d$ that starts at position $i$ in $A$ and position $j$ in~$B$, i.e., $A[i,i+d-1]=B[j,j+d-1]$. The goal of our quantum walk is to find a pair of subsets $(R_1, R_2)$ where $R_1$ and $R_2$ are two subsets of $[n-d+1]$ of size $r$ (for some parameter $r$) such that there exists a marked pair $(i,j)$ in $R_1\times R_2$. Note that since the strings are non-repetitive, for two random subsets $R_1, R_2$, the expected number of pairs $(i,j)\in R_1\times R_2$ such that $A[i]=B[j]$ is roughly $\Theta(r^2/n)$. A simple but crucial observation is that only those pairs can be marked since marked pairs should agree on their first character. Thus for two random subsets $R_1, R_2$ we can check if there exists a marked pair in $R_1\times R_2$ in expected time $\tilde O(\sqrt{r^2/n})$ using Grover search, since we have only $\Theta(r^2/n)$ candidates for marked pairs. This significantly improves over the upper bound $\tilde O(\sqrt{r^2})$ we would get without using the assumption that the strings are non-repetitive. This is how we can improve the complexity down to $\tilde O(n^{2/3}+\sqrt{nd})$. Note that a technical difficulty that we need to overcome is guaranteeing that the running time of the checking procedure is small not only for random subsets but for also all $(R_1,R_2)$, i.e., we need a guarantee on the worst-case running time of the checking procedure. We solve this issue by disregarding the pairs of subsets $(R_1,R_2)$ that contain too many candidates --- we are able to prove that the impact is negligible by using concentration bounds.

The improvement for approximate $\textsf{LCS}$ uses a very similar idea. The main difference is that now we consider a pair $(i,j)\in [n-d+1]\times [n-d+1]$ marked if there is a common substring of length $\ceil{(1-\epsilon)d}$ that starts at position $i$ in $A$ and position $j$ in $B$. Since the fraction of marked pairs increases by a factor $\epsilon d$ we obtain a further improvement. Analyzing the running time of the resulting quantum walk shows that we obtain overall time complexity $\tilde O(n^{2/3}/d^{1/3}+\sqrt{n})$.

\subsection{Ulam distance}
Finally, we present a sublinear time quantum algorithm that computes a $(1+\epsilon)$-approximation of the Ulam distance (i.e., the edit distance for non-repetitive strings).
\begin{theorem}\label{theorem:ulam}
\TheoremUlam
\end{theorem}
In comparison, classical algorithms require linear time even for computing a constant-factor approximation of the Ulam distance when the distance is small (see, e.g., \cite{Andoni+SODA10} for a discussion of the classical lower bounds). Theorem \ref{theorem:ulam} thus shows that while for general strings it is still unknown whether a quantum speed-up is achievable for the computation of  the edit distance, we can obtain a quadratic speed-up for non-repetitive strings. In Section \ref{sec:lb} we show a quantum lower bound (see Theorem \ref{theorem:ulam-approx-LB}) that matches (up to possible polylogarithmic factors) the upper bound of Theorem~\ref{theorem:ulam}. Since it is easy to show that any quantum algorithm that computes the Ulam distance exactly requires $\Omega(n)$ time,
our results are thus essentially optimal.

Let us now describe briefly how the quantum algorithm of Theorem \ref{theorem:ulam} is obtained. Our approach is based on a prior work by Naumovitz, Saks, and Seshadhri \cite{Naumovitz+SODA17} that showed how to construct, for any constant $\epsilon>0$,  a classical algorithm that computes a $(1+\epsilon)$-approximation of the Ulam distance and runs in sublinear time when $\edit(A,B)$ is large (the running time becomes linear when $\edit(A,B)$ is small). The core technique is a variant of the Saks-Seshadhri algorithm for estimating the longest increasing sequence from~\cite{Saks+FOCS10}, which can be used to construct a binary ``indicator'' (we denote this indicator $\textsf{UlamIndic}$ in Section \ref{sec:Ulam}) that outputs $1$ with probability $p$ and $0$ with probability $1-p$, for some value $p$ that is related to the value of $\edit(A,B)$. Conceptually, the approach is based on estimating the value of $\edit(A,B)$ from this indicator using a hierarchy of gap tests and estimators, each with successively better run time. This results in a fairly complex algorithm.

At a high level, our strategy is applying quantum amplitude estimation on the classical indicator $\textsf{UlamIndic}$ to estimate $p$ and thus $\edit(A,B)$. Several technical difficulties nevertheless arise since the indicator requires a rough initial estimation of $\edit(A,B)$ to work efficiently. To solve these difficulties, we first construct a quantum gap test based on quantum amplitude estimation that enables to test efficiently if the success probability of an indicator is larger than some given threshold~$q$ or smaller than $(1-\eta)q$ for some given gap parameter (see Proposition \ref{prop:qtest} in Section~\ref{sec:Ulam}). We then show how to apply this gap test to the indicator \textsf{UlamIndic} with successively better initial estimates of $\edit(A,B)$ in order to obtain a $(1+\epsilon)$-approximation of $\edit(A,B)$ in $\tilde O(\sqrt{n})$ time.
\section{Longest Common Substring}\label{sec:lcs}
Recall that in \textsf{LCS}, we are given two strings $A$ and $B$ of size $n$ and the goal is to find the largest string $t$ which is a substring of both $A$ and $B$. This problem can be solved in time $O(n)$ via suffix tree in the classical setting. In this section, we give upper bounds for the quantum complexity of this problem (lower bounds are discussed in Section \ref{sec:lb}). We first begin by giving an algorithm for exact \textsf{LCS} that runs in time $\tilde O(n^{5/6})$ in Section \ref{sub:lcs-exact}. We then give an algorithm for approximate \textsf{LCS} that runs in time $\tilde O(n^{3/4})$ in Section \ref{sub:lcs-approximate}. 

\subsection{Quantum algorithm for exact \textsf{LCS}}\label{sub:lcs-exact}
In our algorithm, we do a binary search on the size of the solution. We denote this value by $d$. Thus, by losing an $O(\log n)$ factor in the runtime, we reduce the problem to verifying if a substring of size $d$ is shared between the two strings. Our approach is twofold; we design separate algorithms for small and large $d$.

\subsubsection{Quantum algorithm for small $d$}
Our algorithm for small $d$ is based on a reduction to element distinctness. 
Let $|\Sigma|$ be the size of the alphabet and $v:\Sigma \rightarrow [0,|\Sigma|-1]$ be a function that maps every element of the alphabet to a distinct number in range $0 \ldots |\Sigma|-1$. In other words, $v$ is a perfect hashing for the characters. We extend this definition to substrings of the two strings. For a string $t$, we define $v(t)$ as follows:
$$v(t) = \sum_{i = 1}^{|t|} v(t_i) |\Sigma|^{i-1}.$$
It follows from the definition that two strings $t$ and $t’$ are equal if and only if we have $v(t) = v(t’)$. 

From the two strings, we make two sets of numbers $S_A$ and $S_B$ each having $n-d+1$ elements. Element $i$ of set $S_A$ is a pair $(A,i)$ whose value is equal to $v(A[i,i+d-1])$ and similarly element~$i$ of $S_B$ is a pair $(B,i)$ whose value is equal to $v(B[i,i+d-1])$. 
The two sets contain elements with equal values if and only if the two strings have a common substring of size~$d$. Therefore, by solving element distinctness for $S_A$ and $S_B$ we can find out if the size of the solution is at least~$d$. Although element distinctness can be solved in time $\tilde O(n^{2/3})$ when element comparison can be implemented in $\tilde O(1)$ time, our algorithm needs more runtime since comparing elements takes time $\omega(1)$. More precisely, each comparison can be implemented in time $\tilde O(\sqrt{d})$ in the following way: In order to compare two elements of the two sets, we can use Grover's search to find out if the two substrings are different in any position and if so we can find the leftmost position in time $\tilde O(\sqrt{d})$. Thus, it takes time $\tilde O(\sqrt{d})$ to compare two elements which results in runtime $\tilde O(n^{2/3} \sqrt{d})$. We summarize this result in the following lemma.

\begin{lemma}\label{lemma:firstone}
There exists a quantum algorithm that runs in time $\tilde O(n^{2/3}\sqrt{d})$ and verifies with probability $9/10$ if there is a common substring of length $d$ between the two strings.
\end{lemma}

\subsubsection{Quantum algorithm for large $d$}
We now present a quantum algorithm that runs in time $\tilde O(n/\sqrt{d})$ and verifies if there is a common substring of length $d$ between the two strings. 

\paragraph{General description of the algorithm.}
We say that a character of $A$ is marked if it appears among the first $\lfloor d/3 \rfloor$ characters of some substring of length $d$ shared between $A$ and $B$. For example, If there is exactly one common substring of length $d$, there are precisely $\lfloor d/3 \rfloor$ marked characters but we may have more marked characters as the number of common subsequences of length $d$ between $A$ and $B$ increases.
In our algorithm, we sample a substring of length $2\lfloor d/3 \rfloor$ of $A$ and a substring of length $d$ of $B$. 
We call the substring sampled from $A$ the pattern and denote it by $P$ and denote the substring sampled from $B$ by~$S$.  We denote the intervals of $A$ and $B$ that correspond to $P$ and $S$ by $[\ell_P, r_P]$ and $[\ell_S, r_S]$ respectively. That is, $A[\ell_P, r_P] = P$ and $B[\ell_S, r_S] = S$.
We say that the pair $(P,S)$ is good if the following conditions hold (see Figure~\ref{fig:lcs0} for an illustration):
\begin{itemize}
\item There exists a pair $(i,j)$ such that $A[i,i+d-1] = B[j,j+d-1]$ is a common subsequence of size $d$ between the two strings.
\item $i \leq \ell_P$ and $\ell_P - i < \lfloor d/3 \rfloor$.
\item $\ell_S \leq j-i + \ell_P  \leq j-i + r_P \leq r_S$.
\end{itemize}
It directly follows that if $(P,S)$ is a good pair, $\ell_P$ has to be a marked character. When we fix a good pair $(P,S)$ and we refer to the optimal solution, we mean the common subsequence of size $d$ made by $A[i, i+d-1]$ and $B[j, j+d-1]$.

\begin{figure}[!htbp]
	\centering

\tikzset{every picture/.style={line width=0.75pt}} 

\begin{tikzpicture}[x=0.75pt,y=0.75pt,yscale=-1,xscale=1]

\draw   (52,60) -- (607,60) -- (607,91) -- (52,91) -- cycle ;
\draw  [fill={rgb, 255:red, 255; green, 210; blue, 132 }  ,fill opacity=0.5 ][dash pattern={on 0.84pt off 2.51pt}] (219.5,64) -- (420,64) -- (420,87) -- (219.5,87) -- cycle ;
\draw   (52,120) -- (607,120) -- (607,151) -- (52,151) -- cycle ;
\draw  [fill={rgb, 255:red, 189; green, 16; blue, 224 }  ,fill opacity=0.3 ][dash pattern={on 0.84pt off 2.51pt}] (248.5,124) -- (520,124) -- (520,147) -- (248.5,147) -- cycle ;
\draw  [dash pattern={on 0.84pt off 2.51pt}]  (219.5,87) -- (299.5,154) ;
\draw  [dash pattern={on 0.84pt off 2.51pt}]  (419.5,87) -- (500,154) ;
\draw  [fill={rgb, 255:red, 255; green, 210; blue, 132 }  ,fill opacity=0.5 ][dash pattern={on 0.84pt off 2.51pt}] (299.5,154) -- (500,154) -- (500,177) -- (299.5,177) -- cycle ;
\draw    (197.5,43.99) -- (508.5,43.01) ;
\draw [shift={(510.5,43)}, rotate = 539.8199999999999] [color={rgb, 255:red, 0; green, 0; blue, 0 }  ][line width=0.75]    (10.93,-3.29) .. controls (6.95,-1.4) and (3.31,-0.3) .. (0,0) .. controls (3.31,0.3) and (6.95,1.4) .. (10.93,3.29)   ;
\draw [shift={(195.5,44)}, rotate = 359.82] [color={rgb, 255:red, 0; green, 0; blue, 0 }  ][line width=0.75]    (10.93,-3.29) .. controls (6.95,-1.4) and (3.31,-0.3) .. (0,0) .. controls (3.31,0.3) and (6.95,1.4) .. (10.93,3.29)   ;
\draw    (277.5,197.99) -- (588.5,197.01) ;
\draw [shift={(590.5,197)}, rotate = 539.8199999999999] [color={rgb, 255:red, 0; green, 0; blue, 0 }  ][line width=0.75]    (10.93,-3.29) .. controls (6.95,-1.4) and (3.31,-0.3) .. (0,0) .. controls (3.31,0.3) and (6.95,1.4) .. (10.93,3.29)   ;
\draw [shift={(275.5,198)}, rotate = 359.82] [color={rgb, 255:red, 0; green, 0; blue, 0 }  ][line width=0.75]    (10.93,-3.29) .. controls (6.95,-1.4) and (3.31,-0.3) .. (0,0) .. controls (3.31,0.3) and (6.95,1.4) .. (10.93,3.29)   ;

\draw (27,67.4) node [anchor=north west][inner sep=0.75pt]    {$A$};
\draw (27,127.4) node [anchor=north west][inner sep=0.75pt]    {$B$};
\draw (317,67.4) node [anchor=north west][inner sep=0.75pt]    {$P$};
\draw (377,127.4) node [anchor=north west][inner sep=0.75pt]    {$S$};
\draw (245,18) node [anchor=north west][inner sep=0.75pt]   [align=left] {common substring of size $\displaystyle d$};
\draw (335,208) node [anchor=north west][inner sep=0.75pt]   [align=left] {common substring of size $\displaystyle d$};

\end{tikzpicture}

	\caption{An example of a good pair $(P,Q)$. The orange part of $S$ shows the part of $S$ that matches with~$P$.}\label{fig:lcs0}
\end{figure}
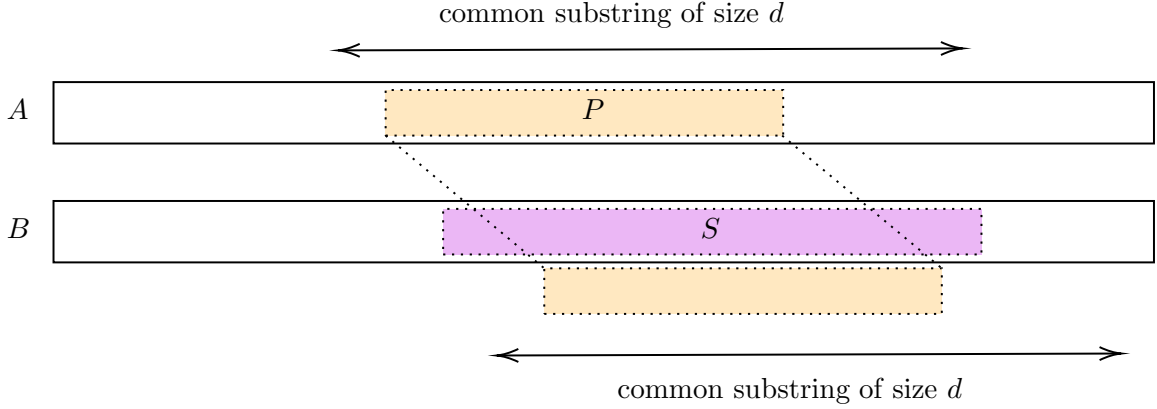

If there is no substring of length $d$ shared between $A$ and $B$, then obviously no good pair exists.
Let us now compute the probability of sampling such $P$ and $S$ under the assumption that there is a common substring of length $d$ shared between $A$ and $B$. There are at least $\lfloor d/3 \rfloor$ marked characters in $A$ and thus with probability $\Omega(d/n)$ we sample the right pattern. Moreover, the corresponding substring of solution in $B$ has at least $\lfloor d/3 \rfloor$ positions such that if we sample~$S$ starting from those positions the pair $(P,S)$ is good. Thus, with probability $\Omega(d^2/n^2)$ we sample a good pair $(P,S)$. 

We describe below a quantum procedure that given a good pair $(P,S)$, constructs with probability at least $1-1/\poly(n)$ a common substring of length $d$ in time $\tilde O(\sqrt{d})$. By first sampling $(P,S)$ and then running this procedure, we get a one-sided procedure that verifies if there exists a common substring of length $d$ shared between $A$ and $B$ with probability $\Omega(d^2/n^2)$ in time $\tilde O(\sqrt{d})$. With amplitude amplification, we can thus construct a quantum algorithm that verifies if there exists a common substring of length $d$ with high probability in time $\tilde O(\sqrt{n^2/d^2}\cdot \sqrt{d})=\tilde O(n/\sqrt{d})$. We summarize this result in the following lemma.

\begin{lemma}\label{lemma:secondone}
There exists a quantum algorithm that runs in time $\tilde O(n/\sqrt{d})$ and verifies with probability $9/10$ if there is a common substring of length $d$ between the two strings.
\end{lemma}

\paragraph*{Constructing a common substring from a good pair.}
From here on, we assume that $(P,S)$ is good and describe how to construct a common substring of length $d$ in time $\tilde O(\sqrt{d})$. 

We aim to find positions of $S$ that match with $P$. For this purpose, we use the string matching algorithm of Ramesh and Vinay~\cite{ramesh2003string} that searches $P$ in $S$ in time $\tilde O(\sqrt{|S|}+\sqrt{|P|})$. If there is no match, then we can conclude that $S$ and $P$ do not meet our property and therefore this should not happen. If there is one such position, then we can detect in time $\tilde O(\sqrt{d})$ if by extending this matching from two ends we can obtain a common substring of size $d$ as follows: We use Grover's search to find the left-most (up to $d$ characters away) position where the two substrings differ when extending them from right. We do the same from the left and it tells us if this gives a common substring of size at least $d$. Each search takes time $\tilde O(\sqrt{d})$. We refer the reader to Figure \ref{fig:lcs1} for a pictorial illustration.

\begin{figure}[!htbp]
	\centering

\tikzset{every picture/.style={line width=0.75pt}} 

\begin{tikzpicture}[x=0.75pt,y=0.75pt,yscale=-1,xscale=1]

\draw   (52,60) -- (607,60) -- (607,91) -- (52,91) -- cycle ;
\draw  [fill={rgb, 255:red, 255; green, 210; blue, 132 }  ,fill opacity=0.5 ][dash pattern={on 0.84pt off 2.51pt}] (219.5,64) -- (420,64) -- (420,87) -- (219.5,87) -- cycle ;
\draw   (52,120) -- (607,120) -- (607,151) -- (52,151) -- cycle ;
\draw  [fill={rgb, 255:red, 189; green, 16; blue, 224 }  ,fill opacity=0.3 ][dash pattern={on 0.84pt off 2.51pt}] (248.5,124) -- (520,124) -- (520,147) -- (248.5,147) -- cycle ;
\draw  [dash pattern={on 0.84pt off 2.51pt}]  (219.5,87) -- (299.5,154) ;
\draw  [dash pattern={on 0.84pt off 2.51pt}]  (419.5,87) -- (500,154) ;
\draw  [fill={rgb, 255:red, 255; green, 210; blue, 132 }  ,fill opacity=0.5 ][dash pattern={on 0.84pt off 2.51pt}] (299.5,154) -- (500,154) -- (500,177) -- (299.5,177) -- cycle ;
\draw  [fill={rgb, 255:red, 208; green, 2; blue, 27 }  ,fill opacity=0.4 ][dash pattern={on 0.84pt off 2.51pt}] (192.5,64) -- (219,64) -- (219,87) -- (192.5,87) -- cycle ;
\draw  [fill={rgb, 255:red, 126; green, 211; blue, 33 }  ,fill opacity=0.4 ][dash pattern={on 0.84pt off 2.51pt}] (419.5,64) -- (511.5,64) -- (511.5,87) -- (419.5,87) -- cycle ;
\draw  [fill={rgb, 255:red, 208; green, 2; blue, 27 }  ,fill opacity=0.4 ][dash pattern={on 0.84pt off 2.51pt}] (272.5,154) -- (299,154) -- (299,177) -- (272.5,177) -- cycle ;
\draw  [fill={rgb, 255:red, 126; green, 211; blue, 33 }  ,fill opacity=0.4 ][dash pattern={on 0.84pt off 2.51pt}] (499.5,154) -- (591.5,154) -- (591.5,177) -- (499.5,177) -- cycle ;
\draw  [dash pattern={on 0.84pt off 2.51pt}]  (193.5,87) -- (273.5,154) ;
\draw  [dash pattern={on 0.84pt off 2.51pt}]  (509.5,87) -- (590,154) ;
\draw    (197.5,43.99) -- (508.5,43.01) ;
\draw [shift={(510.5,43)}, rotate = 539.8199999999999] [color={rgb, 255:red, 0; green, 0; blue, 0 }  ][line width=0.75]    (10.93,-3.29) .. controls (6.95,-1.4) and (3.31,-0.3) .. (0,0) .. controls (3.31,0.3) and (6.95,1.4) .. (10.93,3.29)   ;
\draw [shift={(195.5,44)}, rotate = 359.82] [color={rgb, 255:red, 0; green, 0; blue, 0 }  ][line width=0.75]    (10.93,-3.29) .. controls (6.95,-1.4) and (3.31,-0.3) .. (0,0) .. controls (3.31,0.3) and (6.95,1.4) .. (10.93,3.29)   ;
\draw    (277.5,197.99) -- (588.5,197.01) ;
\draw [shift={(590.5,197)}, rotate = 539.8199999999999] [color={rgb, 255:red, 0; green, 0; blue, 0 }  ][line width=0.75]    (10.93,-3.29) .. controls (6.95,-1.4) and (3.31,-0.3) .. (0,0) .. controls (3.31,0.3) and (6.95,1.4) .. (10.93,3.29)   ;
\draw [shift={(275.5,198)}, rotate = 359.82] [color={rgb, 255:red, 0; green, 0; blue, 0 }  ][line width=0.75]    (10.93,-3.29) .. controls (6.95,-1.4) and (3.31,-0.3) .. (0,0) .. controls (3.31,0.3) and (6.95,1.4) .. (10.93,3.29)   ;

\draw (27,67.4) node [anchor=north west][inner sep=0.75pt]    {$A$};
\draw (27,127.4) node [anchor=north west][inner sep=0.75pt]    {$B$};
\draw (317,67.4) node [anchor=north west][inner sep=0.75pt]    {$P$};
\draw (377,127.4) node [anchor=north west][inner sep=0.75pt]    {$S$};
\draw (245,18) node [anchor=north west][inner sep=0.75pt]   [align=left] {common substring of size $\displaystyle d$};
\draw (335,208) node [anchor=north west][inner sep=0.75pt]   [align=left] {common substring of size $\displaystyle d$};

\end{tikzpicture}

	\caption{If $P$ matches $S$ in only one position, by extending the two ends in the two strings we can find a common substring of length $d$. The orange part of $S$ shows the part of $S$ that matches with~$P$. Also, the red and green parts and extensions from left and right for the matched parts.}\label{fig:lcs1}
\end{figure}
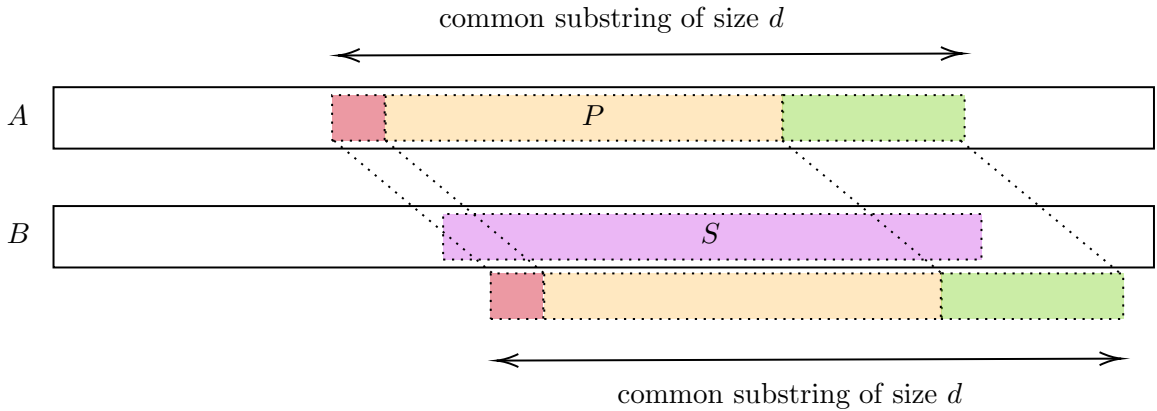

More generally, the above idea works when there are only $O(1)$ many positions of $S$ that match with $P$. However, we should address the case that $P$ appears many times in $S$. In the following we discuss how this can be handled. 
In time $\tilde O(\sqrt{d})$ we first find the leftmost and rightmost positions of $S$ that match with $P$. Let the starting index of the leftmost match be $\ell$ and the starting index of the rightmost match be $r$. In other words, $P = S[\ell, \ell+|P|-1] = S[r,r+|P|-1]$. Since $|S| = d$ and we have $|P| = 2\lfloor d/3 \rfloor$ then the two substrings overlap. Therefore, $P$ as well as the entire string $S[\ell, r+|P|-1]$ is $(r-\ell)$-periodic. This is the key property on which our algorithm will be based.

Let us denote the substring $S[\ell, r+|P|-1]$ by $T$. We extend both $P$ and $T$ from left and right up to a distance of $2d$ in the following way: We increase the index of the ending point so long as the substring remains $(r-\ell)$-periodic. We also stop if we move the ending point by more than $2d$. We do the same for the starting point; we decrease the starting point so long as the substring remains $(r-\ell)$-periodic and up to at most a distance of $2d$. Since we bound the maximum change by $2d$, then this can be done in time $\tilde O(\sqrt{d})$ via Grover's search. Let us denote the resulting substrings by $A[\alpha,\beta]$ and $B[\alpha',\beta']$. The following observation enables us to find the solution in time $\tilde O(\sqrt{d})$: one of the following three cases holds for the optimal solution.
\begin{enumerate}
	\item The starting index of the corresponding optimal solution
	 in $A$ is smaller than $\alpha$. Since the matched parts of the solution are both in the periodic segments, this implies that the first non-periodic indices (when going backwards from the matched parts) in the two solution intervals are $A[\alpha-1]$ and $B[\alpha'-1]$. As a consequence the corresponding character of $A[\alpha]$ in $B$ is $B[\alpha']$ (Figure \ref{fig:three-cases-together}a). 
	\item The ending  index of the corresponding optimal solution in $A$ is larger than $\beta$ and thus (with a similar argument as above) the corresponding character of $A[\beta]$ in $B$ is $B[\beta']$ (Figure \ref{fig:three-cases-together}b). 
	\item The corresponding optimal solution in $A$ is in the interval $A[\alpha,\beta]$ and thus it is $(r-\ell)$-periodic (Figure \ref{fig:three-cases-together}c).
\end{enumerate}
\begin{figure}[!htbp]
	\centering

	\tikzset{every picture/.style={line width=0.75pt}} 
	
	\begin{tikzpicture}[x=0.75pt,y=0.75pt,yscale=-1,xscale=1,scale=0.65]
	\draw   (52,100) -- (607,100) -- (607,131) -- (52,131) -- cycle ;
	\draw  [fill={rgb, 255:red, 255; green, 210; blue, 132 }  ,fill opacity=0.5 ][dash pattern={on 0.84pt off 2.51pt}] (219.5,104) -- (420,104) -- (420,127) -- (219.5,127) -- cycle ;
	\draw   (52,160) -- (607,160) -- (607,191) -- (52,191) -- cycle ;
	\draw  [fill={rgb, 255:red, 126; green, 211; blue, 33 }  ,fill opacity=0.5 ][dash pattern={on 0.84pt off 2.51pt}] (258.5,164) -- (510.5,164) -- (510.5,187) -- (258.5,187) -- cycle ;
	\draw  [fill={rgb, 255:red, 74; green, 144; blue, 226 }  ,fill opacity=0.32 ][dash pattern={on 0.84pt off 2.51pt}] (179.5,64) -- (468.5,64) -- (468.5,87) -- (179.5,87) -- cycle ;
	\draw  [fill={rgb, 255:red, 155; green, 155; blue, 155 }  ,fill opacity=0.5 ][dash pattern={on 0.84pt off 2.51pt}] (249.5,208) -- (591.5,208) -- (591.5,231) -- (249.5,231) -- cycle ;
	\draw  [fill={rgb, 255:red, 208; green, 2; blue, 27 }  ,fill opacity=0.4 ][dash pattern={on 0.84pt off 2.51pt}] (159.5,28) -- (448.5,28) -- (448.5,51) -- (159.5,51) -- cycle ;
	\draw  [fill={rgb, 255:red, 208; green, 2; blue, 27 }  ,fill opacity=0.4 ][dash pattern={on 0.84pt off 2.51pt}] (229.5,238) -- (518.5,238) -- (518.5,261) -- (229.5,261) -- cycle ;
	\draw  [dash pattern={on 0.84pt off 2.51pt}]  (179.5,87) -- (249.5,231) ;
	\draw  [dash pattern={on 4.5pt off 4.5pt}]  (179.5,87) -- (179.5,19) ;
	\draw  [dash pattern={on 4.5pt off 4.5pt}]  (468.5,87) -- (468.5,19) ;
	\draw  [dash pattern={on 4.5pt off 4.5pt}]  (249.5,277) -- (249.5,209) ;
	\draw  [dash pattern={on 4.5pt off 4.5pt}]  (591.5,276) -- (591.5,208) ;
	
	\draw (27,107.4) node [anchor=north west][inner sep=0.75pt]    {$A$};
	\draw (27,167.4) node [anchor=north west][inner sep=0.75pt]    {$B$};
	\draw (317,107.4) node [anchor=north west][inner sep=0.75pt]    {$P$};
	\draw (175,-0.6) node [anchor=north west][inner sep=0.75pt]    {$\alpha $};
	\draw (461,-0.6) node [anchor=north west][inner sep=0.75pt]    {$\beta $};
	\draw (244,278.4) node [anchor=north west][inner sep=0.75pt]    {$\alpha '$};
	\draw (587,274.4) node [anchor=north west][inner sep=0.75pt]    {$\beta '$};
	\end{tikzpicture}
   \caption*{(a) First case.}\vspace{5mm}

\begin{tikzpicture}[x=0.75pt,y=0.75pt,yscale=-1,xscale=1,scale=0.65]
	\draw   (52,100) -- (607,100) -- (607,131) -- (52,131) -- cycle ;
	\draw  [fill={rgb, 255:red, 255; green, 210; blue, 132 }  ,fill opacity=0.5 ][dash pattern={on 0.84pt off 2.51pt}] (219.5,104) -- (420,104) -- (420,127) -- (219.5,127) -- cycle ;
	\draw   (52,160) -- (607,160) -- (607,191) -- (52,191) -- cycle ;
	\draw  [fill={rgb, 255:red, 126; green, 211; blue, 33 }  ,fill opacity=0.5 ][dash pattern={on 0.84pt off 2.51pt}] (258.5,164) -- (510.5,164) -- (510.5,187) -- (258.5,187) -- cycle ;
	\draw  [fill={rgb, 255:red, 74; green, 144; blue, 226 }  ,fill opacity=0.32 ][dash pattern={on 0.84pt off 2.51pt}] (179.5,64) -- (468.5,64) -- (468.5,87) -- (179.5,87) -- cycle ;
	\draw  [fill={rgb, 255:red, 155; green, 155; blue, 155 }  ,fill opacity=0.5 ][dash pattern={on 0.84pt off 2.51pt}] (249.5,208) -- (543.5,208) -- (543.5,231) -- (249.5,231) -- cycle ;
	\draw  [fill={rgb, 255:red, 208; green, 2; blue, 27 }  ,fill opacity=0.4 ][dash pattern={on 0.84pt off 2.51pt}] (207.5,28) -- (506.5,28) -- (506.5,51) -- (207.5,51) -- cycle ;
	\draw  [fill={rgb, 255:red, 208; green, 2; blue, 27 }  ,fill opacity=0.4 ][dash pattern={on 0.84pt off 2.51pt}] (287.5,237) -- (576.5,237) -- (576.5,260) -- (287.5,260) -- cycle ;
	\draw  [dash pattern={on 0.84pt off 2.51pt}]  (468.5,87) -- (543.5,231) ;
	\draw  [dash pattern={on 4.5pt off 4.5pt}]  (179.5,87) -- (179.5,19) ;
	\draw  [dash pattern={on 4.5pt off 4.5pt}]  (468.5,87) -- (468.5,19) ;
	\draw  [dash pattern={on 4.5pt off 4.5pt}]  (249.5,277) -- (249.5,209) ;
	\draw  [dash pattern={on 4.5pt off 4.5pt}]  (543.5,276) -- (543.5,208) ;
	
	\draw (27,107.4) node [anchor=north west][inner sep=0.75pt]    {$A$};
	\draw (27,167.4) node [anchor=north west][inner sep=0.75pt]    {$B$};
	\draw (317,107.4) node [anchor=north west][inner sep=0.75pt]    {$P$};
	\draw (175,-0.6) node [anchor=north west][inner sep=0.75pt]    {$\alpha $};
	\draw (461,-0.6) node [anchor=north west][inner sep=0.75pt]    {$\beta $};
	\draw (244,278.4) node [anchor=north west][inner sep=0.75pt]    {$\alpha '$};
	\draw (539,274.4) node [anchor=north west][inner sep=0.75pt]    {$\beta '$};
	\end{tikzpicture}
   \caption*{(b) Second case.}\vspace{5mm}

\begin{tikzpicture}[x=0.75pt,y=0.75pt,yscale=-1,xscale=1,scale=0.65]
	\draw   (52,100) -- (607,100) -- (607,131) -- (52,131) -- cycle ;
	\draw  [fill={rgb, 255:red, 255; green, 210; blue, 132 }  ,fill opacity=0.5 ][dash pattern={on 0.84pt off 2.51pt}] (219.5,104) -- (420,104) -- (420,127) -- (219.5,127) -- cycle ;
	\draw   (52,160) -- (607,160) -- (607,191) -- (52,191) -- cycle ;
	\draw  [fill={rgb, 255:red, 126; green, 211; blue, 33 }  ,fill opacity=0.5 ][dash pattern={on 0.84pt off 2.51pt}] (258.5,164) -- (510.5,164) -- (510.5,187) -- (258.5,187) -- cycle ;
	\draw  [fill={rgb, 255:red, 74; green, 144; blue, 226 }  ,fill opacity=0.32 ][dash pattern={on 0.84pt off 2.51pt}] (179.5,64) -- (468.5,64) -- (468.5,87) -- (179.5,87) -- cycle ;
	\draw  [fill={rgb, 255:red, 155; green, 155; blue, 155 }  ,fill opacity=0.5 ][dash pattern={on 0.84pt off 2.51pt}] (249.5,208) -- (543.5,208) -- (543.5,231) -- (249.5,231) -- cycle ;
	\draw  [fill={rgb, 255:red, 208; green, 2; blue, 27 }  ,fill opacity=0.4 ][dash pattern={on 0.84pt off 2.51pt}] (207.5,28) -- (446.5,28) -- (446.5,51) -- (207.5,51) -- cycle ;
	\draw  [dash pattern={on 4.5pt off 4.5pt}]  (179.5,87) -- (179.5,19) ;
	\draw  [dash pattern={on 4.5pt off 4.5pt}]  (468.5,87) -- (468.5,19) ;
	\draw  [dash pattern={on 4.5pt off 4.5pt}]  (249.5,277) -- (249.5,209) ;
	\draw  [dash pattern={on 4.5pt off 4.5pt}]  (543.5,276) -- (543.5,208) ;
	\draw  [fill={rgb, 255:red, 208; green, 2; blue, 27 }  ,fill opacity=0.4 ][dash pattern={on 0.84pt off 2.51pt}] (287.5,248) -- (526.5,248) -- (526.5,271) -- (287.5,271) -- cycle ;
	
	\draw (27,107.4) node [anchor=north west][inner sep=0.75pt]    {$A$};
	\draw (27,167.4) node [anchor=north west][inner sep=0.75pt]    {$B$};
	\draw (317,107.4) node [anchor=north west][inner sep=0.75pt]    {$P$};
	\draw (175,-0.6) node [anchor=north west][inner sep=0.75pt]    {$\alpha $};
	\draw (461,-0.6) node [anchor=north west][inner sep=0.75pt]    {$\beta $};
	\draw (244,278.4) node [anchor=north west][inner sep=0.75pt]    {$\alpha '$};
	\draw (539,274.4) node [anchor=north west][inner sep=0.75pt]    {$\beta '$};
	\end{tikzpicture}
	  \caption*{(c) Third case.}

	\caption{The three cases considered. The red intervals correspond to the longest common substring. The yellow interval is $P$ and the green interval is $S[\ell,r+|P|-1]$. The blue and grey intervals are extensions of the yellow and green intervals so long as they remain $(r-\ell)$-periodic.}\label{fig:three-cases-together}
\end{figure}

In all three cases, we can find the optimal solution in time $\tilde O(\sqrt{d})$. For the first two cases, we know one correspondence between the two common substrings of length $d$. More precisely, in Case~1 we know that $A[\alpha]$ is part of the solution and this character corresponds to $B[\alpha']$. Thus, it suffices to do a Grover's search from two ends to extend the matching in the two directions. Similarly in Case 2 we know that $A[\beta]$ corresponds to $B[\beta']$ and thus via a Grover's search we find a common substring of length $d$ in the two strings.

Finally, for Case 3 we point out that since the entire solution lies in intervals $A[\alpha,\beta]$ and $B[\alpha', \beta']$ then we have $\beta \geq \alpha + d-1$ and $\beta' \geq \alpha' + d-1$. Notice that both intervals $A[\alpha,\beta]$ and $B[\alpha',\beta']$ are $(r-\ell)$-periodic. Let the starting position of $P$ in $A$ be $x$. For a fixed position $y$ in the interval $[\alpha',\beta']$ that matches with $P$, if we extend this matching from the two ends we obtain a solution of size $\min\{x-\alpha,y-\alpha'\} + \min\{\beta-x+1, \beta'-y+1\}$. This means that extending one of the following matchings between $P$ and $S$ gives us a common subsequence of size $d$:

\begin{itemize}
	\item An optimal solution when $x-\alpha \geq y-\alpha'$: the rightmost matching in the interval $B[\alpha',\alpha'+|P|+(x-\alpha)]$, or
	\item An optimal solution when $x-\alpha \leq y-\alpha'$: the leftmost matching in the interval $B[\alpha'+(x-\alpha),\beta']$.
\end{itemize}

Each matching can be found in time $\tilde O(\sqrt{d})$ and similar to the ideas explained above we can extend each matching to verify if it gives us a common substring of size $d$ in time $\tilde O(\sqrt{d})$. 

\subsubsection{Combining the two algorithms}
Combining Lemmas \ref{lemma:firstone} and \ref{lemma:secondone} gives us a solution in time $\tilde O(n^{5/6})$.
\begin{theorem-repeat}{theorem:lcs-exact}
\TheoremLCSExact
\end{theorem-repeat}
\begin{proof}
	We do a binary search on $d$ (size of the solution). To verify a given $d$, we consider two cases: if $d < n^{1/3}$ we run the algorithm of Lemma \ref{lemma:firstone} and otherwise we run the algorithm of Lemma~\ref{lemma:secondone}. Thus, the overall runtime is bounded by $\tilde O(n^{5/6})$.
\end{proof}

\subsection{Quantum algorithm for approximate \textsf{LCS}}\label{sub:lcs-approximate}
In the following, we show that if an approximate solution is desired then we can improve the runtime down to $\tilde O(n^{3/4})$. Similar to the above discussion, we use two algorithms for small $d$ and large~$d$. Our algorithm for large $d$ is the same as the one we use for exact solution. For small $d$ we modify the algorithm of Lemma \ref{lemma:firstone} to improve its runtime down to $\tilde O(n^{2/3}d^{1/6})$. This is explained in Lemma~\ref{lemma:3}.

\begin{lemma}\label{lemma:3}
For any constant $0 < \epsilon < 1$, there exists a quantum algorithm that runs in time $\tilde O(n^{2/3}d^{1/6})$ and if the two strings share a common substring of length $d$, finds a common substring of length $(1-\epsilon)d$.
\end{lemma}
\begin{proof}
	Similar to Lemma \ref{lemma:firstone}, we use element distinctness for this algorithm. We make two sets $S_A$ and $S_B$ and prove that they share two equal elements if and only if their corresponding substrings of size $(1-\epsilon)d$ are equal. The difference between this algorithm and the algorithm of Lemma \ref{lemma:firstone} is that here $S_A$ and $S_B$ contain $O(n/\sqrt{\epsilon d})$ elements instead of $n-d+1$ elements. 

Let $k =  \sqrt{\epsilon d}$. We break both strings into blocks of size $k$ and make the two sets in the following way (see Figure~\ref{fig:f4} for an illustration):
	\begin{itemize}
		\item For each $i$ such that $i \textsf{ mod } k = 1$, we put element $(A,i)$ in set $S_A$.
		\item For each $i$ such that $\lceil i / k \rceil\textsf{ mod } k = 1$ we put element $(B,i)$ in set $S_B$.
	\end{itemize}
	\begin{figure}[!htbp]
	\centering

\tikzset{every picture/.style={line width=0.75pt}} 

\begin{tikzpicture}[x=0.75pt,y=0.75pt,yscale=-1,xscale=1]

\draw   (52,20) -- (607,20) -- (607,51) -- (52,51) -- cycle ;
\draw   (52,80) -- (607,80) -- (607,111) -- (52,111) -- cycle ;
\draw    (53.5,127) -- (107.5,127) ;
\draw [shift={(109.5,127)}, rotate = 180] [color={rgb, 255:red, 0; green, 0; blue, 0 }  ][line width=0.75]    (10.93,-3.29) .. controls (6.95,-1.4) and (3.31,-0.3) .. (0,0) .. controls (3.31,0.3) and (6.95,1.4) .. (10.93,3.29)   ;
\draw [shift={(51.5,127)}, rotate = 0] [color={rgb, 255:red, 0; green, 0; blue, 0 }  ][line width=0.75]    (10.93,-3.29) .. controls (6.95,-1.4) and (3.31,-0.3) .. (0,0) .. controls (3.31,0.3) and (6.95,1.4) .. (10.93,3.29)   ;
\draw    (113.5,127) -- (167.5,127) ;
\draw [shift={(169.5,127)}, rotate = 180] [color={rgb, 255:red, 0; green, 0; blue, 0 }  ][line width=0.75]    (10.93,-3.29) .. controls (6.95,-1.4) and (3.31,-0.3) .. (0,0) .. controls (3.31,0.3) and (6.95,1.4) .. (10.93,3.29)   ;
\draw [shift={(111.5,127)}, rotate = 0] [color={rgb, 255:red, 0; green, 0; blue, 0 }  ][line width=0.75]    (10.93,-3.29) .. controls (6.95,-1.4) and (3.31,-0.3) .. (0,0) .. controls (3.31,0.3) and (6.95,1.4) .. (10.93,3.29)   ;
\draw    (173.5,127) -- (227.5,127) ;
\draw [shift={(229.5,127)}, rotate = 180] [color={rgb, 255:red, 0; green, 0; blue, 0 }  ][line width=0.75]    (10.93,-3.29) .. controls (6.95,-1.4) and (3.31,-0.3) .. (0,0) .. controls (3.31,0.3) and (6.95,1.4) .. (10.93,3.29)   ;
\draw [shift={(171.5,127)}, rotate = 0] [color={rgb, 255:red, 0; green, 0; blue, 0 }  ][line width=0.75]    (10.93,-3.29) .. controls (6.95,-1.4) and (3.31,-0.3) .. (0,0) .. controls (3.31,0.3) and (6.95,1.4) .. (10.93,3.29)   ;
\draw    (273.5,126) -- (327.5,126) ;
\draw [shift={(329.5,126)}, rotate = 180] [color={rgb, 255:red, 0; green, 0; blue, 0 }  ][line width=0.75]    (10.93,-3.29) .. controls (6.95,-1.4) and (3.31,-0.3) .. (0,0) .. controls (3.31,0.3) and (6.95,1.4) .. (10.93,3.29)   ;
\draw [shift={(271.5,126)}, rotate = 0] [color={rgb, 255:red, 0; green, 0; blue, 0 }  ][line width=0.75]    (10.93,-3.29) .. controls (6.95,-1.4) and (3.31,-0.3) .. (0,0) .. controls (3.31,0.3) and (6.95,1.4) .. (10.93,3.29)   ;
\draw  [fill={rgb, 255:red, 208; green, 2; blue, 27 }  ,fill opacity=0.4 ][dash pattern={on 0.84pt off 2.51pt}] (56,24) -- (65.5,24) -- (65.5,47) -- (56,47) -- cycle ;
\draw  [fill={rgb, 255:red, 208; green, 2; blue, 27 }  ,fill opacity=0.4 ][dash pattern={on 0.84pt off 2.51pt}] (116,24) -- (125.5,24) -- (125.5,47) -- (116,47) -- cycle ;
\draw  [fill={rgb, 255:red, 208; green, 2; blue, 27 }  ,fill opacity=0.4 ][dash pattern={on 0.84pt off 2.51pt}] (176,24) -- (185.5,24) -- (185.5,47) -- (176,47) -- cycle ;
\draw  [fill={rgb, 255:red, 208; green, 2; blue, 27 }  ,fill opacity=0.4 ][dash pattern={on 0.84pt off 2.51pt}] (276,24) -- (285.5,24) -- (285.5,47) -- (276,47) -- cycle ;
\draw  [fill={rgb, 255:red, 208; green, 2; blue, 27 }  ,fill opacity=0.4 ][dash pattern={on 0.84pt off 2.51pt}] (56,84) -- (111.5,84) -- (111.5,107) -- (56,107) -- cycle ;
\draw  [fill={rgb, 255:red, 208; green, 2; blue, 27 }  ,fill opacity=0.4 ][dash pattern={on 0.84pt off 2.51pt}] (274,84) -- (329.5,84) -- (329.5,107) -- (274,107) -- cycle ;

\draw (27,27.4) node [anchor=north west][inner sep=0.75pt]    {$A$};
\draw (27,87.4) node [anchor=north west][inner sep=0.75pt]    {$B$};
\draw (72,129.4) node [anchor=north west][inner sep=0.75pt]    {$k$};
\draw (134,129.4) node [anchor=north west][inner sep=0.75pt]    {$k$};
\draw (194,129.4) node [anchor=north west][inner sep=0.75pt]    {$k$};
\draw (238,117.4) node [anchor=north west][inner sep=0.75pt]    {$\dotsc $};
\draw (294,128.4) node [anchor=north west][inner sep=0.75pt]    {$k$};
\draw (338,117.4) node [anchor=north west][inner sep=0.75pt]    {$\dotsc $};

\end{tikzpicture}

	\caption{$k = \sqrt{\epsilon d}$. Only the elements colored in red are included in the two sets.}\label{fig:f4}
\end{figure}
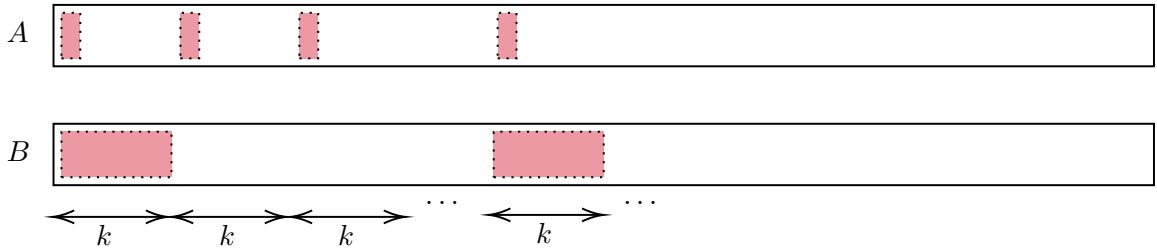
	It follows that among the $\epsilon d$ first characters of the solution, there is one position which is included in $S_A$ such that its corresponding position in $B$ is also included in $S_B$. Thus, if we define the equality between two elements $(A,i)$ and $(B,j)$ as $A[i,i+(1-\epsilon)d-1] = B[j,j+(1-\epsilon)d-1] $ then by solving element distinctness for $S_A$ and $S_B$ we can find a desired solution. Since we have $O(n/\sqrt{d})$ elements in each set then the overall runtime is bounded by $\tilde O((n/\sqrt{d})^{2/3}\sqrt{d}) = \tilde O(n^{2/3}d^{1/6})$.
\end{proof}
Lemmas \ref{lemma:secondone} and \ref{lemma:3} 
lead to a quantum algorithm for \textsf{LCS} with approximation factor $1-\epsilon$ and runtime $\tilde O(n^{3/4})$.
\begin{theorem-repeat}{theorem:lcs-approx}
\TheoremLCSApprox
\end{theorem-repeat}
\begin{proof}
	We do a binary search on $d$. To verify a given $d$, we consider two cases: if $d < n^{1/2}$ we run the algorithm of Lemma \ref{lemma:3} and otherwise we run the algorithm of Lemma~\ref{lemma:secondone}. Thus, the overall runtime is bounded by $\tilde O(n^{3/4})$.
\end{proof}
\newcommand{\rr}{r}
\renewcommand{\Cc}{\mathcal{C}}
\renewcommand{\Tt}{\mathcal{T}}

\section{Longest Common Substring for Non-Repetitive Strings}
In this section, we consider \textsf{LCS} for input strings $A$ and $B$ that are non-repetitive (e.g., permutations). We show that there exists an $\tilde O(n^{3/4})$-time quantum algorithm that computes an exact $\textsf{LCS}$ and an $\tilde O(n^{2/3})$-time quantum algorithm that computes a $(1-\epsilon)$-approximation of $\textsf{LCS}$. This significantly improves the generic results of the previous section. 

\subsection{Quantum walks}\label{sub:walks}
We first review the quantum walk approach that solves element distinctness and several similar problems. We actually consider a variant of the standard approach that is especially well suited for our purpose: instead of considering a quantum walk on a Johnson graph we consider a quantum walk on the direct product of two Johnson graphs. While quantum walks on such graphs have also been used by Buhrman and {\v{S}}palek in \cite{buhrman2006quantum}, their implementation used Szegegy's version of quantum walks \cite{szegedy2004quantum}. One noteworthy aspect of our approach is that we rely on the stronger version of quantum walks developed by Magniez, Nayak, Roland and Santha \cite{Magniez+SICOMP11} (this is crucial for the walk we present in Section \ref{sub:LCS-perm-exact} and \ref{sub:LCS-perm-approx}; using Szegegy's version does not lead to an efficient algorithm because the checking cost of our quantum walk is too high).

Let $S$ be a finite set and $r$ be any integer such that $r\le |S|$. Let us denote by $\Tt_r$ the set of all couples $(R_1,R_2)$ where both $R_1$ and $R_2$ are subsets of $S$ of size $r$. The elements of $\Tt_r$ are called the states of the walk. Let $\Tt_r^\ast\subseteq \Tt_r$ be some subset. The elements of $\Tt_r^\ast$ are called the marked states. 
The quantum walk approach associates to each $(R_1,R_2)\in \Tt_r$ a data structure $D(R_1,R_2)$. Three types of costs are associated with~$D$.
The setup cost $s(\rr)$ is the cost to set up the data structure $D(R_1,R_2)$ for any $(R_1,R_2)\in \Tt_r$.
The update cost $u(\rr)$ is the cost to update $D(R_1,R_2)$ for any $(R_1,R_2)\in \Tt_r$: converting $D(R_1,R_2)$ to $D((R_1\setminus{a})\cup \{a'\},(R_2\setminus{b})\cup \{b'\})$ for any $a\in R_1$, $a'\in S\backslash R_1$, $b\in R_2$ and $b'\in S\backslash R_2$. 
The checking cost $c(\rr)$ is the cost of checking, given $D(R_1,R_2)$ for any $(R_1,R_2)\in \Tt_r$, if $(R_1,R_2)\in \Tt_r^\ast$. 
By combining amplitude amplification and quantum walks, we can find with high probability one marked state $(R_1,R_2)$.
We will use in this paper the following statement (we refer to \cite{Magniez+SICOMP11} for details) of this approach. 

\begin{proposition}\label{fact_magniez}
Consider any value $\rr\in\{1,\ldots,|S|\}$ and any constant $\gamma>0$.
Assume that either $\Tt_r^\ast=\emptyset$ or $\frac{|\Tt_r^\ast|}{|\Tt_r|}\ge \delta$ holds for some known value $\delta>0$.
There exists a quantum algorithm that always rejects if $\Tt_r^\ast=\emptyset$, 
outputs a set $(R_1,R_2)\in \Tt_r^\ast$ with probability at least $1-\frac{1}{n^\gamma}$ otherwise, and has complexity 
\[
\tilde O\left(s(\rr)+\sqrt{1/\delta}\left(\sqrt{\rr}\times u(\rr)+c(\rr)\right)\right).
\]
\end{proposition} 

Let us now explain how to solve element distinctness using this framework. Let $x_1,\ldots,x_n$ denote the $n$ elements of the first list, and $y_1,\ldots,y_n$ denote the $n$ elements of the second list of the input of element distinctness. Let us first discuss the complexity in the standard access model, where we assume that each element is encoded using $O(\log n)$ bits and that given $i\in[n]$ we can obtain all the bits of $x_i$ and $y_i$ in $O(\log n)$ time.
We set $S=\{1,\ldots, n\}$ and say that a state $(R_1,R_2)\in \Tt_r$ is marked if there exists a pair $(i,j)\in R_1\times R_2$ such that $x_i=y_j$. Easy calculations show that the fraction of marked states is $\Omega(r^2/n^2)$. By using an appropriate data structure that allows insertion, deletion and lookup operations in polylogarithmic time (see Section~6.2 of \cite{ambainis2007quantum} for details about how to construct such a data structure) to store the elements~$x_i$ for all $i\in R_1$, and using another instance of this data structure to store the elements $y_j$ for all $j\in R_2$, we can perform the setup operation in time $s(r)=\tilde O (r)$ and perform each update operation in time $u(r)=\tilde O(1)$. Each checking operation can be implemented in time $c(r)=\tilde O(\sqrt{r})$  as follows: perform a Grover search over $R_1$ to check if there exists $i\in R_1$ for which $x_i=y_j$ for some $j\in R_2$ (once $i$ is fixed the later property can be checked in polylogarithmic time using the data structure representing $R_2$).\footnote{The checking operation can actually be implemented more efficiently but the cost $c(r)=\tilde O(\sqrt{\rr})$ is enough for our purpose.}
From Proposition~\ref{fact_magniez}, the overall time complexity of the quantum walk is thus
\[
\tilde O\left(r+\frac{n}{r}\left(\sqrt{\rr}\times 1+\sqrt{\rr}\right)\right).
\]
By taking $r=n^{2/3}$ we get time complexity $\tilde O(n^{2/3})$. As discussed in  Section 6.1 of \cite{ambainis2007quantum}, the above data structure, and thus the whole quantum walk can also be implemented in the same way in the (weaker) comparison model: if given any pair of indices $(i,j)\in[n]\times [n]$ we can decide in $T(n)$ time whether $x_i<y_j$ or $x_i\ge y_j$, then the time complexity of the implementation is $\tilde O(n^{2/3}T(n))$.

\subsection{Quantum algorithm for exact \textsf{LCS} of non-repetitive strings}\label{sub:LCS-perm-exact}
We set $S=\{1,\ldots,n-d+1\}$. Let us write $\alpha=54\log n$. 
For any state $(R_1,R_2)\in \Tt_r$ we define the data structure $D(R_1,R_2)$ as follows:
\begin{itemize}
\item
$D(R_1,R_2)$ records all the values $A[i]$ and $B[j]$ for all $i\in R_1$ and $j\in R_2$ using the same data structure as the data structure considered above for element distinctness;
\item
additionally, $D(R_1,R_2)$ records the number of pairs $(i,j)\in R_1\times R_2$ such that $A[i]=B[j]$, and stores explicitly all these pairs (using a history-independent data structure updatable in polylogarithmic time, which can be constructed based on the data structure from \cite{ambainis2007quantum}).
\end{itemize}
We define the set of marked states $\Tt_r^\ast$ as follows. We say that $(R_1,R_2)$ is marked if the following two conditions hold:
\begin{itemize}
\item[(i)]
there exists a pair $(i,j)\in R_1\times R_2$ such that $A[i,i+d-1]=B[j,j+d-1]$;
\item[(ii)]
the number of pairs $(i,j)\in R_1\times R_2$ such that $A[i]=B[j]$ is at most $\ceil{\alpha (r^2/n+1)}$.
\end{itemize}
The following lemma is easy to show.
\begin{lemma}\label{lemma:ratio}
If the two non-repetitive strings $A$ and $B$ have a common substring of length~$d$, then the fraction of marked states is 
$\Omega\left(r^2/n^2\right)$.
\end{lemma}
\begin{proof}
The fraction of states $(R_1,R_2)$ for which there exists a pair $(i,j)\in R_1\times R_2$ such that $A[i, i+d-1]=B[j, j+d-1]$ is at least 
\[
\frac{{n-2\choose r-2}}{{n \choose r}}=\frac{r(r-1)}{n(n-1)}.
\]
The fraction of $(R_1,R_2)$ such that Condition (i) holds is thus $\Omega(r^2/n^2)$.

Let $X$ be the random variable representing the number of pairs $(i,j)\in R_1\times R_2$ such that $A[i]=B[j]$ when $(R_1,R_2)$ is taken uniformly at random in $\Tt_r$. We will show that $X\le \ceil{\alpha (r^2/n+1)}$ with high probability. Let us prove this upper bound for the worst possible case: input strings for which the number of pairs $(i,j)\in S\times S$ such that $A[i]=B[j]$ is $|S|$.


Assume that we have fixed $R_1$. For each $i\in R_1$ there exists a unique index $j\in S$ such that $A[i]=B[j]$. Let $Y$ be the random variable representing the total number of pairs $(i,j)\in R_1\times R_2$ such that $A[i]=B[j]$ when choosing $R_2$. Note that $Y$ follows an hypergeometric distribution of mean $r^2/n$. From standard extensions of Chernoff’s bound (see, e.g., Section 1.6 of \cite{Dubhashi+09}), we get 
\begin{align*}
\Pr\left[Y \ge  \alpha\left( \frac{r^2}{n}+1\right)\right]\le \exp\left(-(3\log n)\cdot \left(\frac{r^2}{n}+1\right)\right)\le \frac{1}{n^3}.
\end{align*}
Thus the fraction of $(R_1,R_2)$ such that Condition (ii) does not hold is at most $1/n^3$.

The statement of the lemma then follows from the union bound.
%
%
%
\end{proof}

Let us now analyze the costs corresponding to this quantum walk. The setup and update operations are similar to the corresponding operations of the quantum walk for element distinctness described in Section \ref{sub:walks}. The only difference is that we need to keep track of the pairs $(i,j)\in R_1\times R_2$ such that $A[i]=B[j]$. Note that for each $i\in [n]$, there is at most one index $j\in R_2$ such that $A[i]=B[j]$, since the strings are non-repetitive. Moreover this index can be found in polylogarithmic time using the data structure associated with $R_2$. Similarly, for each $j\in [n]$, there is at most one index $i\in R_1$ such that $A[i]=B[j]$, which can be also found in polylogarithmic time. Thus the time complexities of the setup and update operations 
 are $s(r) = \tilde O(r)$ and $u(r)=\tilde O(1)$, respectively, as in the quantum walk for element distinctness described in Section \ref{sub:walks}.
The checking operation first checks if the number of pairs $(i,j)\in R_1\times R_2$ such that $A[i]=B[j]$ is at most $\ceil{\alpha (r^2/n+1)}$ and then, if this condition is satisfied, performs a Grover search on all pairs $(i,j)\in R_1\times R_2$ such that $A[i]=B[j]$ stored in $D(R_1,R_2)$, in order to check if there exists a pair $(i,j)\in R_1\times R_2$ such that $A[i, i+d-1]=B[j, j+d-1]$. The time complexity is
\[
c(r) = \tilde O(1+ \sqrt{\ceil{\alpha (r^2/n+1)}}\cdot \sqrt{d}).
\]
Using Proposition~\ref{fact_magniez}, we get that the overall time complexity of the quantum walk is 
\[
\tilde O\left(r+\frac{n}{r}(\sqrt{r}\times 1 + \sqrt{\ceil{\alpha (r^2/n+1)}}\cdot \sqrt{d})\right)=\tilde O\left(r+\frac{n}{\sqrt{r}}+\frac{n\sqrt{d}}{r}+\sqrt{nd}\right).
\]
For $d\le n^{1/3}$, this expression is minimized for $r=n^{2/3}$ and gives complexity $\tilde O(n^{2/3})$. For $d\ge n^{1/3}$, the complexity is dominated by the last term $\tilde O(\sqrt{nd})$.

By combining the above algorithm with the $\tilde O(n/\sqrt{d})$-time algorithm of Lemma \ref{lemma:secondone}, we get overall complexity $\tilde O(n^{3/4})$.
\begin{theorem-repeat}{theorem:lcs-exact-perm}
\TheoremLCSExactPerm
\end{theorem-repeat}
\begin{proof}
	We do a binary search on $d$. To verify a given $d$, we consider two cases: if $d < \sqrt{n}$ we run the $\tilde O(\sqrt{nd})$-time quantum algorithm we just described and otherwise we run the algorithm of Lemma~\ref{lemma:secondone}. Thus, the overall runtime is bounded by $\tilde O(n^{3/4})$.
\end{proof}

\subsection{Quantum algorithm for approximate \textsf{LCS} of non-repetitive strings}\label{sub:LCS-perm-approx}
We set again $S=\{1,\ldots,n-d+1\}$ and, for any $(R_1,R_2)\in \Tt_r$, we define the data structure $D(R_1,R_2)$ exactly as above.
This time, we say that $(R_1,R_2)$ is marked if the following two conditions hold:
\begin{itemize}
\item[(i)]
there exists a pair $(i,j)\in R_1\times R_2$ such that $A[i, i+\ceil{(1-\epsilon)d}-1]=B[j, j+\ceil{(1-\epsilon)d}-1]$;
\item[(ii)]
the number of pairs $(i,j)\in R_1\times R_2$ such that $A[i]=B[j]$ is at most $\ceil{\alpha r^2/n}$;
\end{itemize}
where we use the same value $\alpha=54\log n$ as above.
We now show the following lemma.
\begin{lemma}\label{lemma:ratio-approx}
Assume that 
$r\le n/d$. 
If $A$ and $B$ have a common subsequence of length~$d$, then the fraction of marked sets is
$\Omega\left(\frac{dr^2}{n^2}\right)$.
\end{lemma}
\begin{proof}
Let us write $m=d-\ceil{(1-\epsilon) d}$ and say that a pair $(i,j)\in S\times S$ is good if $A[i, i+\ceil{(1-\epsilon)d}-1]=B[j, j+\ceil{(1-\epsilon)d}-1]$. There are at least $m$ good pairs in $S\times S$. 
The probability that $R_1$ does not contain any element involved in a good pair is
\begin{align*}
\frac{{n-m\choose r}}{{n \choose r}}&= \frac{(n-m)!r!(n-r)!}{r!(n-m-r)!n!}\\
&=\frac{(n-m)(n-m-1)\cdots (n-m-r+1)}{n(n-1)\cdots(n-r+1)}\\
&\le \left(\frac{n-m}{n}\right)^r=\left(1-\frac{m}{n}\right)^r\le \exp\left(-\frac{rm}{n}\right)\\
&\le 1-\frac{rm}{2n},
\end{align*}
where we used the fact that $\exp(-x)\le 1-x/2$ on $x\in[0,1]$ to derive the last inequality. Thus with probability at least $rm/(2n)$, the set $R_1$ contains at least one element involved in a good pair.

Assuming that $R_1$ contains at least one element involved in a good pair, when choosing $R_2$ at random
with probability at least $r/n$ it contains a good pair.
Thus the overall probability that $R_1\times R_2$ contains a good pair (and thus satisfies Condition (i)) is at least $\Omega(dr^2/n^2)$.

The fraction of $(R_1,R_2)$ such that Condition (ii) does not hold is at most $1/n^3$ (the proof of this claim is exactly the same as for Lemma \ref{lemma:ratio}). The statement of the lemma then follows from the union bound.
\end{proof}

We thus get a quantum walk with running time 
\[
\tilde O\left(r+\frac{n}{r\sqrt{d}}(\sqrt{r}\times 1 + \sqrt{\ceil{\alpha (r^2/n+1)}}\cdot \sqrt{d})\right)
=\tilde O\left(r+\frac{n}{\sqrt{rd}}+\frac{n}{r}+\sqrt{n}\right).
\]
If $d\le \sqrt{n}$ then the expression is minimized for $r=n^{2/3}/d^{1/3}$ and the complexity is $\tilde O(n^{2/3}/d^{1/3})$. (Note that with this value of $r$ the condition of Lemma \ref{lemma:ratio-approx} is satisfied.) If $d> \sqrt{n}$ we obtain complexity $\tilde O(d)$ by taking $r=n/d$.\footnote{For $d> \sqrt{n}$, we can actually obtain the better upper bound $\tilde O(\sqrt{n})$ by slightly modifying the algorithm. The bound  $\tilde O(d)$ is nevertheless sufficient for our purpose.}

By combining the above algorithm with the $\tilde O(n/\sqrt{d})$-time algorithm of Lemma \ref{lemma:secondone}, we get overall complexity $\tilde O(n^{2/3})$.
\begin{theorem-repeat}{theorem:lcs-approx-perm}
\TheoremLCSApproxPerm
\end{theorem-repeat}
\begin{proof}
	We do a binary search on $d$. To verify a given $d$, we consider two cases: if $d < n^{2/3}$ we run the $\tilde O(n^{2/3}/d^{1/3}+d)$-time quantum algorithm we just described and otherwise we run the algorithm of Lemma~\ref{lemma:secondone}. Thus, the overall runtime is bounded by $\tilde O(n^{2/3})$.
\end{proof}
\section{Longest Palindrome Substring}\label{section:lps}
In this section, we present a quantum algorithm that solves the longest palindrome substring problem (\textsf{LPS}) in time $\tilde O(\sqrt{n})$. In Section \ref{sec:lb}, we complement this result by a tight lower bound of $\tilde \Omega(\sqrt{n})$ for quantum algorithms even for the special case of 0/1 strings and even for approximating the solution.

\subsection{General description of our quantum algorithm}
Our algorithm starts with a binary search on the size of the solution. We use $l$ and $u$ as a lower bound and an upper bound on the size of the solution. Initially we set $l = 1$ and $u = n$. Each time we set $d = \lceil (l+u)/2 \rceil$ and try to find a solution of size at least $d$. If such a solution exists, then we continue on by setting $l = d$. Otherwise we set $u = d-1$. After $O(\log n)$ steps we have $l = u$ which is equal to the size of the solution. 
Thus, in the following, we assume that an integer $d$ is given and our goal is to find out if a solution of size at least $d$ exists in the string. 

One thing to keep in mind is that if the solution size is larger than $d$, we may not necessarily have a palindrome substring of size $d$ but in that case, we certainly have a palindrome substring of size $d+1$. Thus, all it takes to verify if the solution size is at least $d$ is to look for a solution of size either $d$ or $d+1$, and we can simply consider each case separately. From here on, we assume that our string contains a palindrome substring of size exactly $d$ and we wish to find such a substring.

For each palindrome substring of length $d$, we mark the first $\lfloor d/2 \rfloor$ characters of this substring. We will show in Section \ref{sub:check} how to check if a character is marked or not in time $\tilde O(\sqrt{d})$. Since we may have several palindrome substrings of length $d$ we may have more than $d/2$ marked characters in our string. However, since there is at least one solution of size $d$, then the number of marked characters is at least $\Omega(d)$. Our quantum algorithm applies Grover algorithm to find one marked character. The running time is $\tilde O(\sqrt{n/d}\cdot\sqrt{d})=\tilde O(\sqrt{n})$. We thus obtain the following result.

\begin{theorem-repeat}{theorem:lps}
\TheoremLPS
\end{theorem-repeat}

\subsection{Description of the checking procedure}\label{sub:check}
We now describe the main idea of the procedure that checks whether a given position $r$ is a marked character or not. This is the main technical contribution of the whole section. 

Define the center of a palindrome substring as the average of its starting position and its ending position (it is not integer if the length of the palindrome substring is even). Note that if we fix a center $c$, we can via Grover's search determine in time $\tilde O(\sqrt{d})$ whether the size of the corresponding palindrome substring is $d$ or not. 

We now bring a crucial observation on which the checking procedure is built. Assume that position~$r$ of the input string is marked. Let $\ell$ be the leftmost character of a corresponding solution, i.e., the substring $A[\ell, \ell+d-1]$ of the input array is palindrome and position $r$ is one of the first $\lfloor d/2 \rfloor$ characters of this interval. Since $A[\ell, \ell+d-1]$ is palindrome then we have that 
$A[r,r+\lceil d/2 \rceil-1]$
is equal to the reverse of $A[2\ell+d-\lceil d/2 \rceil-r , 2\ell+d-r-1]$. See Figure \ref{fig:lcs-observation} for an illustration.

\begin{figure}[!htbp]
	\centering

\tikzset{every picture/.style={line width=0.75pt}} 

\begin{tikzpicture}[x=0.75pt,y=0.75pt,yscale=-1,xscale=1]

\draw   (52,40) -- (607,40) -- (607,71) -- (52,71) -- cycle ;
\draw  [fill={rgb, 255:red, 255; green, 210; blue, 132 }  ,fill opacity=1 ][dash pattern={on 0.84pt off 2.51pt}] (148,44) -- (420,44) -- (420,67) -- (148,67) -- cycle ;
\draw  [fill={rgb, 255:red, 80; green, 227; blue, 194 }  ,fill opacity=0.32 ][dash pattern={on 0.84pt off 2.51pt}] (237,76) -- (509,76) -- (509,99) -- (237,99) -- cycle ;
\draw    (65,59) .. controls (53,300) and (617,298) .. (595,58) ;
\draw [shift={(595,58)}, rotate = 264.76] [color={rgb, 255:red, 0; green, 0; blue, 0 }  ][fill={rgb, 255:red, 0; green, 0; blue, 0 }  ][line width=0.75]      (0, 0) circle [x radius= 3.35, y radius= 3.35]   ;
\draw [shift={(65,59)}, rotate = 92.85] [color={rgb, 255:red, 0; green, 0; blue, 0 }  ][fill={rgb, 255:red, 0; green, 0; blue, 0 }  ][line width=0.75]      (0, 0) circle [x radius= 3.35, y radius= 3.35]   ;
\draw    (95,59) .. controls (85,273) and (588,259) .. (569,60) ;
\draw [shift={(569,60)}, rotate = 264.55] [color={rgb, 255:red, 0; green, 0; blue, 0 }  ][fill={rgb, 255:red, 0; green, 0; blue, 0 }  ][line width=0.75]      (0, 0) circle [x radius= 3.35, y radius= 3.35]   ;
\draw [shift={(95,59)}, rotate = 92.68] [color={rgb, 255:red, 0; green, 0; blue, 0 }  ][fill={rgb, 255:red, 0; green, 0; blue, 0 }  ][line width=0.75]      (0, 0) circle [x radius= 3.35, y radius= 3.35]   ;
\draw    (165,59) .. controls (61,230) and (649,197) .. (501,61) ;
\draw [shift={(501,61)}, rotate = 222.58] [color={rgb, 255:red, 0; green, 0; blue, 0 }  ][fill={rgb, 255:red, 0; green, 0; blue, 0 }  ][line width=0.75]      (0, 0) circle [x radius= 3.35, y radius= 3.35]   ;
\draw [shift={(165,59)}, rotate = 121.31] [color={rgb, 255:red, 0; green, 0; blue, 0 }  ][fill={rgb, 255:red, 0; green, 0; blue, 0 }  ][line width=0.75]      (0, 0) circle [x radius= 3.35, y radius= 3.35]   ;

\draw (54,19.4) node [anchor=north west][inner sep=0.75pt]    {$\ell $};
\draw (545,19.4) node [anchor=north west][inner sep=0.75pt]    {$\ell +d-1$};
\draw (148,19.4) node [anchor=north west][inner sep=0.75pt]    {$r$};
\draw (327,18.4) node [anchor=north west][inner sep=0.75pt]    {$r+\lceil d/2\rceil -1$};
\draw (414,108.4) node [anchor=north west][inner sep=0.75pt]    {$2\ell +d-r-1$};
\draw (236,108.4) node [anchor=north west][inner sep=0.75pt]    {$2\ell +d-\lceil d/2\rceil -r$};
\draw (103,54.4) node [anchor=north west][inner sep=0.75pt]  [font=\Huge]  {$\dotsc $};
\draw (512,54.4) node [anchor=north west][inner sep=0.75pt]  [font=\Huge]  {$\dotsc $};
\draw (286,54.4) node [anchor=north west][inner sep=0.75pt]  [font=\Huge]  {$\dotsc $};

\end{tikzpicture}\vspace{-10mm}

	\caption{If $A[\ell,\ell+d-1]$ is a palindrome, then the blue substring is equal to the reverse of the orange substring.}\label{fig:lcs-observation}
\end{figure}
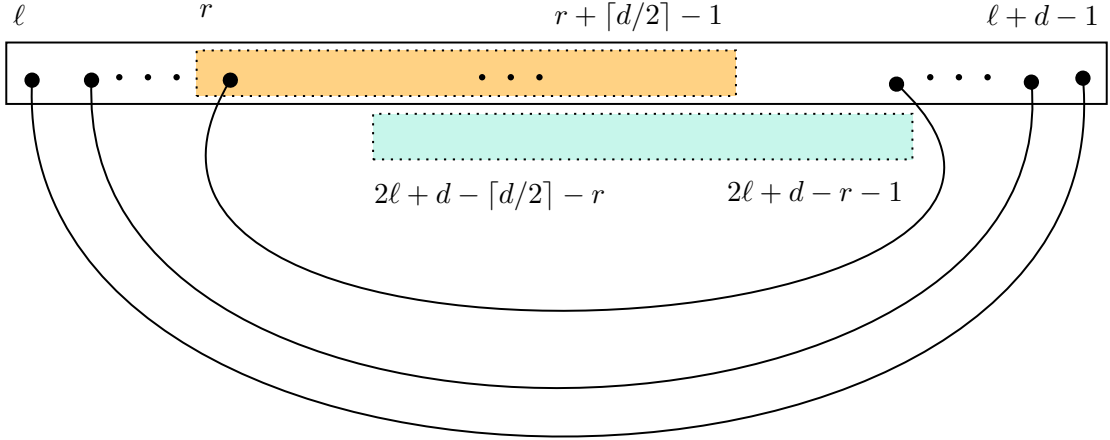

Let $S$ be the substring of length~$d$ that spans the interval $[r,r+d-1]$ of the input string. The main strategy of our checking procedure is to identify positions of $S$ that are candidates for the center of the solution. For this purpose we also define a pattern $P$ of length $\lceil d/2 \rceil$ which is equal to the reverse of interval $[r, r+\lceil d/2 \rceil -1]$. We say a position $i$ of $S$ matches with pattern $P$ if $S[i,i+|P|-1] = P$.

If $r$ is marked, then it is in the left half of some palindrome substring $A[\ell, \ell+d-1]$. Moreover, the size of~$S$ is equal to $d$ and by the way we construct $P$ it is equal to the reverse of $A[r,r+\lceil d/2 \rceil-1]$. Thus, $A[2\ell+d-\lceil d/2 \rceil-r , 2\ell+d-r-1]$ is equal to $P$ and completely lies inside $S$. Thus, by finding all positions of $S$ that match with $P$, we definitely find $A[2\ell+d-\lceil d/2 \rceil-r , 2\ell+d-r-1]$. In particular, if there is only one position of $S$ that matches with $P$, this enables to identify the center of $A[\ell,\ell+d-1]$.

In Section \ref{subsub:pattern} below we show how to identify efficiently all the positions of $S$ that match with pattern $P$. These positions give a list of candidates for the center of a solution. At least one of them is the center of a solution. If there is only one candidate (or a constant number of candidates), we can then via Grover's search determine in time $\tilde O(\sqrt{d})$ whether the size of the corresponding palindrome substring is $d$ or not. In Section~\ref{subsub:check}, we explain how to deal with the case where there are many candidates for the center.

\subsubsection{Identifying the patterns}\label{subsub:pattern}

Below we describe how to identify all positions of the $S$ that match with pattern $P$ by considering the periodicity of $P$. We show that this can be done in time $\tilde O(\sqrt{d})$.

To do so, we first find the rightmost occurrence of the $P$ in $S$ using the quantum algorithm of Ramesh and Vinay~\cite{ramesh2003string}. This takes time $\tilde O(\sqrt{d})$ since $|S|+|P| = O(d)$. If $P$ does not appear in~$S$ at all or only once, we are done. Otherwise, let $o_r$ be the rightmost position of $S$ that matches with $P$. That is $S[o_r,o_r+|P|-1] = P$. Moreover, let $o$ be the position of  the second rightmost occurrence of $P$ in $S$. We argue that $T = \{ o_r - i(o_r - o) \:|\: 0 \leq i \leq \lfloor (o_r-1) / (o_r-o) \rfloor \}$ is the set of all occurrences of $P$ in $S$. More precisely, every element of $T$ is the starting position of an occurrence of $P$ in $S$. Since $|S| \leq 2|P|$, there is one special case in which $o = o_r - |P|$. In this case $|T|$ contains exactly two elements and the proof is trivial. Thus, we assume in the following that the two rightmost occurrences of $P$ in $S$ overlap. 

Remember that we say a string $s$ is $q$-periodic if we have $s_i = s_{i+q}$ for all $1 \leq i \leq |s|-q$. Since~$P$ appears as a substring of $S$ at both locations $o_r$ and $o$, the pattern $P$ is $(o_r-o)$-periodic. Moreover, $P$ is not $i$-periodic for any $1 < i < o_r - o$ since otherwise $o$ would not be the position of the second rightmost occurrence of $P$ in $S$. We first prove that each element of $T$ is a position in $S$ that matches with $P$. Recall that by the way we construct $S$ and $P$, $P$ is equal to the reverse of the first $|P|$ characters of $S$. Moreover, $|S| \leq 2|P|$ and thus the entire interval $S[1,o_r+|P|-1]$ is $(o_r-o)$-periodic (we refer to Figure \ref{fig:periodic} for an illustration). 
This implies that every element of $T$ is the beginning of one occurrence of $P$ in $S$. 

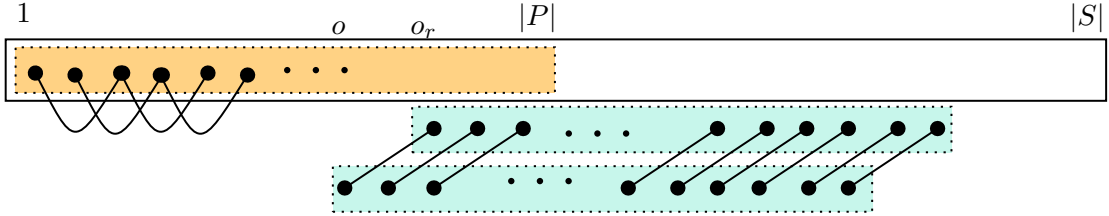
\begin{figure}[!htbp]
	\centering

	\tikzset{every picture/.style={line width=0.75pt}} 

\tikzset{every picture/.style={line width=0.75pt}} 

\begin{tikzpicture}[x=0.75pt,y=0.75pt,yscale=-1,xscale=1]

\draw   (52,40) -- (607,40) -- (607,71) -- (52,71) -- cycle ;
\draw  [fill={rgb, 255:red, 255; green, 210; blue, 132 }  ,fill opacity=1 ][dash pattern={on 0.84pt off 2.51pt}] (57,44) -- (329,44) -- (329,67) -- (57,67) -- cycle ;
\draw  [fill={rgb, 255:red, 80; green, 227; blue, 194 }  ,fill opacity=0.32 ][dash pattern={on 0.84pt off 2.51pt}] (257,74) -- (529,74) -- (529,97) -- (257,97) -- cycle ;
\draw  [fill={rgb, 255:red, 80; green, 227; blue, 194 }  ,fill opacity=0.32 ][dash pattern={on 0.84pt off 2.51pt}] (217,104) -- (489,104) -- (489,127) -- (217,127) -- cycle ;
\draw    (522,85) -- (477,115) ;
\draw [shift={(477,115)}, rotate = 146.31] [color={rgb, 255:red, 0; green, 0; blue, 0 }  ][fill={rgb, 255:red, 0; green, 0; blue, 0 }  ][line width=0.75]      (0, 0) circle [x radius= 3.35, y radius= 3.35]   ;
\draw [shift={(522,85)}, rotate = 146.31] [color={rgb, 255:red, 0; green, 0; blue, 0 }  ][fill={rgb, 255:red, 0; green, 0; blue, 0 }  ][line width=0.75]      (0, 0) circle [x radius= 3.35, y radius= 3.35]   ;
\draw    (502,85) -- (457,115) ;
\draw [shift={(457,115)}, rotate = 146.31] [color={rgb, 255:red, 0; green, 0; blue, 0 }  ][fill={rgb, 255:red, 0; green, 0; blue, 0 }  ][line width=0.75]      (0, 0) circle [x radius= 3.35, y radius= 3.35]   ;
\draw [shift={(502,85)}, rotate = 146.31] [color={rgb, 255:red, 0; green, 0; blue, 0 }  ][fill={rgb, 255:red, 0; green, 0; blue, 0 }  ][line width=0.75]      (0, 0) circle [x radius= 3.35, y radius= 3.35]   ;
\draw    (477,85) -- (432,115) ;
\draw [shift={(432,115)}, rotate = 146.31] [color={rgb, 255:red, 0; green, 0; blue, 0 }  ][fill={rgb, 255:red, 0; green, 0; blue, 0 }  ][line width=0.75]      (0, 0) circle [x radius= 3.35, y radius= 3.35]   ;
\draw [shift={(477,85)}, rotate = 146.31] [color={rgb, 255:red, 0; green, 0; blue, 0 }  ][fill={rgb, 255:red, 0; green, 0; blue, 0 }  ][line width=0.75]      (0, 0) circle [x radius= 3.35, y radius= 3.35]   ;
\draw    (456,85) -- (411,115) ;
\draw [shift={(411,115)}, rotate = 146.31] [color={rgb, 255:red, 0; green, 0; blue, 0 }  ][fill={rgb, 255:red, 0; green, 0; blue, 0 }  ][line width=0.75]      (0, 0) circle [x radius= 3.35, y radius= 3.35]   ;
\draw [shift={(456,85)}, rotate = 146.31] [color={rgb, 255:red, 0; green, 0; blue, 0 }  ][fill={rgb, 255:red, 0; green, 0; blue, 0 }  ][line width=0.75]      (0, 0) circle [x radius= 3.35, y radius= 3.35]   ;
\draw    (436,85) -- (391,115) ;
\draw [shift={(391,115)}, rotate = 146.31] [color={rgb, 255:red, 0; green, 0; blue, 0 }  ][fill={rgb, 255:red, 0; green, 0; blue, 0 }  ][line width=0.75]      (0, 0) circle [x radius= 3.35, y radius= 3.35]   ;
\draw [shift={(436,85)}, rotate = 146.31] [color={rgb, 255:red, 0; green, 0; blue, 0 }  ][fill={rgb, 255:red, 0; green, 0; blue, 0 }  ][line width=0.75]      (0, 0) circle [x radius= 3.35, y radius= 3.35]   ;
\draw    (411,85) -- (366,115) ;
\draw [shift={(366,115)}, rotate = 146.31] [color={rgb, 255:red, 0; green, 0; blue, 0 }  ][fill={rgb, 255:red, 0; green, 0; blue, 0 }  ][line width=0.75]      (0, 0) circle [x radius= 3.35, y radius= 3.35]   ;
\draw [shift={(411,85)}, rotate = 146.31] [color={rgb, 255:red, 0; green, 0; blue, 0 }  ][fill={rgb, 255:red, 0; green, 0; blue, 0 }  ][line width=0.75]      (0, 0) circle [x radius= 3.35, y radius= 3.35]   ;
\draw    (313,85) -- (268,115) ;
\draw [shift={(268,115)}, rotate = 146.31] [color={rgb, 255:red, 0; green, 0; blue, 0 }  ][fill={rgb, 255:red, 0; green, 0; blue, 0 }  ][line width=0.75]      (0, 0) circle [x radius= 3.35, y radius= 3.35]   ;
\draw [shift={(313,85)}, rotate = 146.31] [color={rgb, 255:red, 0; green, 0; blue, 0 }  ][fill={rgb, 255:red, 0; green, 0; blue, 0 }  ][line width=0.75]      (0, 0) circle [x radius= 3.35, y radius= 3.35]   ;
\draw    (290,85) -- (245,115) ;
\draw [shift={(245,115)}, rotate = 146.31] [color={rgb, 255:red, 0; green, 0; blue, 0 }  ][fill={rgb, 255:red, 0; green, 0; blue, 0 }  ][line width=0.75]      (0, 0) circle [x radius= 3.35, y radius= 3.35]   ;
\draw [shift={(290,85)}, rotate = 146.31] [color={rgb, 255:red, 0; green, 0; blue, 0 }  ][fill={rgb, 255:red, 0; green, 0; blue, 0 }  ][line width=0.75]      (0, 0) circle [x radius= 3.35, y radius= 3.35]   ;
\draw    (268,85) -- (223,115) ;
\draw [shift={(223,115)}, rotate = 146.31] [color={rgb, 255:red, 0; green, 0; blue, 0 }  ][fill={rgb, 255:red, 0; green, 0; blue, 0 }  ][line width=0.75]      (0, 0) circle [x radius= 3.35, y radius= 3.35]   ;
\draw [shift={(268,85)}, rotate = 146.31] [color={rgb, 255:red, 0; green, 0; blue, 0 }  ][fill={rgb, 255:red, 0; green, 0; blue, 0 }  ][line width=0.75]      (0, 0) circle [x radius= 3.35, y radius= 3.35]   ;
\draw    (67,57) .. controls (87,96) and (86,97) .. (111,57) ;
\draw [shift={(111,57)}, rotate = 302.01] [color={rgb, 255:red, 0; green, 0; blue, 0 }  ][fill={rgb, 255:red, 0; green, 0; blue, 0 }  ][line width=0.75]      (0, 0) circle [x radius= 3.35, y radius= 3.35]   ;
\draw [shift={(67,57)}, rotate = 62.85] [color={rgb, 255:red, 0; green, 0; blue, 0 }  ][fill={rgb, 255:red, 0; green, 0; blue, 0 }  ][line width=0.75]      (0, 0) circle [x radius= 3.35, y radius= 3.35]   ;
\draw    (87,58) .. controls (107,97) and (106,98) .. (131,58) ;
\draw [shift={(131,58)}, rotate = 302.01] [color={rgb, 255:red, 0; green, 0; blue, 0 }  ][fill={rgb, 255:red, 0; green, 0; blue, 0 }  ][line width=0.75]      (0, 0) circle [x radius= 3.35, y radius= 3.35]   ;
\draw [shift={(87,58)}, rotate = 62.85] [color={rgb, 255:red, 0; green, 0; blue, 0 }  ][fill={rgb, 255:red, 0; green, 0; blue, 0 }  ][line width=0.75]      (0, 0) circle [x radius= 3.35, y radius= 3.35]   ;
\draw    (110,57) .. controls (130,96) and (129,97) .. (154,57) ;
\draw [shift={(154,57)}, rotate = 302.01] [color={rgb, 255:red, 0; green, 0; blue, 0 }  ][fill={rgb, 255:red, 0; green, 0; blue, 0 }  ][line width=0.75]      (0, 0) circle [x radius= 3.35, y radius= 3.35]   ;
\draw [shift={(110,57)}, rotate = 62.85] [color={rgb, 255:red, 0; green, 0; blue, 0 }  ][fill={rgb, 255:red, 0; green, 0; blue, 0 }  ][line width=0.75]      (0, 0) circle [x radius= 3.35, y radius= 3.35]   ;
\draw    (130,58) .. controls (150,97) and (149,98) .. (174,58) ;
\draw [shift={(174,58)}, rotate = 302.01] [color={rgb, 255:red, 0; green, 0; blue, 0 }  ][fill={rgb, 255:red, 0; green, 0; blue, 0 }  ][line width=0.75]      (0, 0) circle [x radius= 3.35, y radius= 3.35]   ;
\draw [shift={(130,58)}, rotate = 62.85] [color={rgb, 255:red, 0; green, 0; blue, 0 }  ][fill={rgb, 255:red, 0; green, 0; blue, 0 }  ][line width=0.75]      (0, 0) circle [x radius= 3.35, y radius= 3.35]   ;

\draw (56,20.4) node [anchor=north west][inner sep=0.75pt]    {$1$};
\draw (309,20.4) node [anchor=north west][inner sep=0.75pt]    {$|P|$};
\draw (248,20.4) node [anchor=north west][inner sep=0.75pt]    {$ \begin{array}{l}
	o_{r}\\
	\end{array}$};
ßß

\draw (208,20.4) node [anchor=north west][inner sep=0.75pt]    {$ \begin{array}{l}
	o\\
	\end{array}$};
\draw (330,84.4) node [anchor=north west][inner sep=0.75pt]  [font=\Huge]  {$\dotsc $};
\draw (301,108.4) node [anchor=north west][inner sep=0.75pt]  [font=\Huge]  {$\dotsc $};
\draw (188,52.4) node [anchor=north west][inner sep=0.75pt]  [font=\Huge]  {$\dotsc $};
\draw (587,20.4) node [anchor=north west][inner sep=0.75pt]    {$|S|$};

\end{tikzpicture}

	\caption{The green substring is equal to pattern $P$ and the orange substring is the reverse of $P$. Since $P$ is $(o_r-o)$-periodic, then the entire substring of $S[1,o_r+|P|-1]$ is also $(o_r-o)$-periodic. }\label{fig:periodic}
\end{figure}

In order to show that $T$ contains all occurrences of $P$ in $S$, assume for the sake of contradiction that there is a position $p$ of $S$ that is the beginning of one occurrence of $P$ but $p$ is not included in $T$. This implies that there is one element $e \in T$ such that $|e-p| < o_r - o$. Thus, with the same argument $P$ is $|e-p|$-periodic. This is in contradiction with what we proved previously. 

 \subsubsection{Checking the patterns}\label{subsub:check}


We have shown that the set $T$ computed in Section \ref{subsub:pattern}
is the set of the starting positions of each occurrence of $P$ in $S$. We define 
\[
C = \{ r-1 + (e + \lceil d/2 \rceil)/2  \:|\: e \in T\},
\]
which is the set of all possible centers for the solution.  Since we consider all occurrences of $P$ in $S$, one of the elements in $C$ is equal to the actual center of our solution. If $|C| = O(1)$, one can iterate over all elements of $C$ and verify in time $O(\sqrt{d})$ if each one is a center for a palindrome substring of size $d$. 
In the following, we show that even if $|C|$ is large, we can narrow down the search to a constant number of elements in $C$, which makes it possible to find efficiently a palindrome substring of length~$d$.

We showed in Section \ref{subsub:pattern} that $C$ is large only if $P$ is periodic with small periodicity. Let $\alpha$ be the periodicity of~$P$. As we discuss in Section \ref{subsub:pattern}, every consecutive pair of elements in $T$ have distance $\alpha$.  Thus, each pair of consecutive elements in $C$ have distance $\alpha/2$. Let $y$ be the largest index of $A$ such that $A[r,y]$ is $\alpha$-periodic. Also, let $x$ be the smallest index of $A$ such that $A[x,y]$ is $\alpha$-periodic. We prove two things in the following: 
\begin{itemize}
\item[1)]
If $y-x+1 \geq d+\alpha$ 
then we already have a palindrome substring of size $d$ in the interval $[x,y]$ that can easily be found in time $\tilde O(\sqrt{d})$;
\item[2)] 
Otherwise the center of our solution has a distance of at most $\alpha/2$ from $(x+y)/2$, and thus we only need to consider a constant number of elements in $C$, which enables us to solve the problem in time $\tilde O(\sqrt{d})$ as well.
\end{itemize}
Here is the proof of the first statement.
\begin{lemma}\label{lemma-lcs1}
	If $y-x+1 \geq d+\alpha$ then one can find a palindrome substring of size at least $d$ in time $\tilde O(\sqrt{d})$.
\end{lemma}
\begin{proof}
	
		\begin{figure}[!htbp]
	\centering

\tikzset{every picture/.style={line width=0.75pt}} 

\begin{tikzpicture}[x=0.75pt,y=0.75pt,yscale=-1,xscale=1]

\draw   (26,140) -- (640,140) -- (640,171) -- (26,171) -- cycle ;
\draw  [fill={rgb, 255:red, 255; green, 210; blue, 132 }  ,fill opacity=1 ][dash pattern={on 0.84pt off 2.51pt}] (451,144) -- (573,144) -- (573,167) -- (451,167) -- cycle ;
\draw  [fill={rgb, 255:red, 255; green, 210; blue, 132 }  ,fill opacity=1 ][dash pattern={on 0.84pt off 2.51pt}] (241,144) -- (363,144) -- (363,167) -- (241,167) -- cycle ;
\draw  [fill={rgb, 255:red, 126; green, 211; blue, 33 }  ,fill opacity=0.47 ][dash pattern={on 0.84pt off 2.51pt}] (26,174) -- (600,174) -- (600,197) -- (26,197) -- cycle ;
\draw    (605,180) -- (635,180) ;
\draw [shift={(637,180)}, rotate = 180] [color={rgb, 255:red, 0; green, 0; blue, 0 }  ][line width=0.75]    (10.93,-3.29) .. controls (6.95,-1.4) and (3.31,-0.3) .. (0,0) .. controls (3.31,0.3) and (6.95,1.4) .. (10.93,3.29)   ;
\draw [shift={(603,180)}, rotate = 0] [color={rgb, 255:red, 0; green, 0; blue, 0 }  ][line width=0.75]    (10.93,-3.29) .. controls (6.95,-1.4) and (3.31,-0.3) .. (0,0) .. controls (3.31,0.3) and (6.95,1.4) .. (10.93,3.29)   ;
\draw    (252.5,111.65) .. controls (294.33,74.7) and (412.11,77.64) .. (449.36,112.39) ;
\draw [shift={(451,114)}, rotate = 225.78] [fill={rgb, 255:red, 0; green, 0; blue, 0 }  ][line width=0.08]  [draw opacity=0] (10.72,-5.15) -- (0,0) -- (10.72,5.15) -- (7.12,0) -- cycle    ;
\draw [shift={(250,114)}, rotate = 315] [fill={rgb, 255:red, 0; green, 0; blue, 0 }  ][line width=0.08]  [draw opacity=0] (10.72,-5.15) -- (0,0) -- (10.72,5.15) -- (7.12,0) -- cycle    ;

\draw (26,120.4) node [anchor=north west][inner sep=0.75pt]    {$x$};
\draw (517,118.4) node [anchor=north west][inner sep=0.75pt]    {$ \begin{array}{l}
	r+e+|P|-2\\
	\end{array}$};
\draw (447,118.4) node [anchor=north west][inner sep=0.75pt]    {$ \begin{array}{l}
	r\\
	\end{array}$};
\draw (627,120.4) node [anchor=north west][inner sep=0.75pt]    {$y$};
\draw (238,118.4) node [anchor=north west][inner sep=0.75pt]    {$x'$};
\draw (348,118.4) node [anchor=north west][inner sep=0.75pt]    {$x'$};
\draw (26,210.4) node [anchor=north west][inner sep=0.75pt]    {$x''$};
\draw (587,210.4) node [anchor=north west][inner sep=0.75pt]    {$y''$};
\draw (600,184.4) node [anchor=north west][inner sep=0.75pt]    {$\ < \alpha $};
\draw (304,60) node [anchor=north west][inner sep=0.75pt]   [align=left] {distance};
\draw (359,59.4) node [anchor=north west][inner sep=0.75pt]    {$\equiv 0$};
\draw (360,52) node [anchor=north west][inner sep=0.75pt]    {$\alpha $};

\end{tikzpicture}

	\caption{Illustration for the proof of Lemma \ref{lemma-lcs1}. All the colored intervals are palindrome.}
\end{figure}
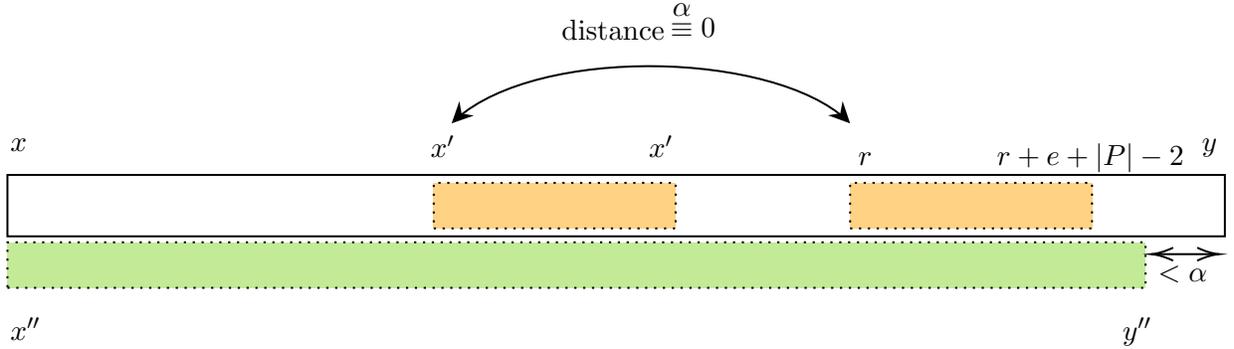
	Let $e \in T$ be an arbitrary element of $T$. This means that $P = A[r+e-1,r+e-1+|P|-1]$. Since $A[r,r+|P|-1]$ is equal to the reverse of $P$ and the two intervals $[r,r+|P|-1]$ and $[r+e-1,r+e-1+|P|-1]$ overlap then $A[r,r+e-1+|P|-1]$ is a palindrome substring. Recall that $P$ is $\alpha$-periodic and thus so is the entire substring $A[r,r+e-1+|P|-1]$. This implies that interval $[r,r+e-1+|P|-1]$ is inside interval $[x,y]$. We start by $x' = r$ and $y' = r+e-1+|P|-1$. Next, we shift $[x',y']$ within $[x,y]$ by multiples of $\alpha$ until $||x-x'|-|y-y'|| \leq \alpha$. Since $A[x,y]$ is $\alpha$-periodic, we still have $A[x',y'] = A[r,r+e-1+|P|-1]$ and therefore $A[x',y']$ is palindrome. Finally, we set $x'' = x' - \min\{|x-x'|,|y-y'|\}$ and $y'' = y'+ \min\{|x-x'|,|y-y'|\}$. It follows from the periodicity of $A[x,y]$ that $A[x'',y'']$ is palindrome. Moreover, the size of $[x'',y'']$ is at most $\alpha$ smaller that the size of $[x,y]$. Thus, $A[x'',y'']$ is a desired substring.

	It is easy to find such a solution in time $\tilde O(\sqrt{d})$. The only part of the proof which is non-trivial from a computational standpoint is finding $[x,y]$. If $y-x = O(d)$ then one can find it via Grover's search in time $\tilde O(\sqrt{d})$. Otherwise, instead of finding $x$ and $y$, we extend $r$ from both ends up to a distance of $2d$ so long as it remains $\alpha$-periodic. This can be done in time $\tilde O(\sqrt{d})$ and then with the same analysis we can find a palindrome substring of size at least $d$.
\end{proof}

The proof of the second statement relies on the following lemma.

\begin{lemma}\label{lemma:2}
   If $y-x+1 < d+\alpha$ then the center of the solution cannot be more than $2\alpha$ away from $(x+y)/2$.
\end{lemma}
\begin{proof}
	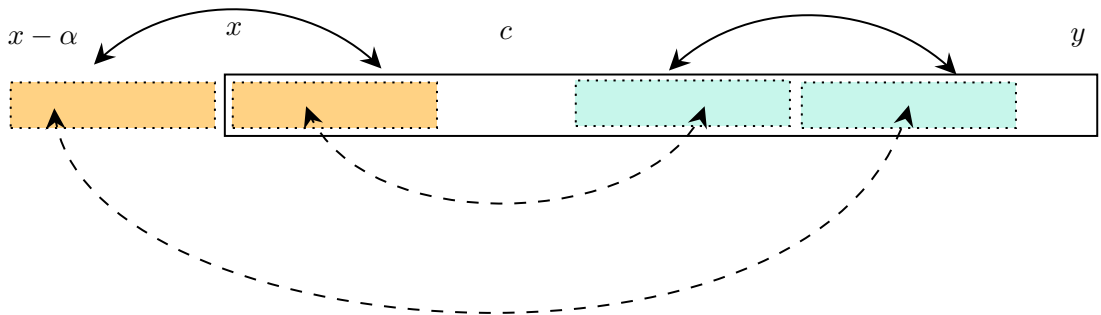
\begin{figure}[!htbp]
	\centering

\tikzset{every picture/.style={line width=0.75pt}} 

\begin{tikzpicture}[x=0.75pt,y=0.75pt,yscale=-1,xscale=1]

\draw   (175,86) -- (615,86) -- (615,117) -- (175,117) -- cycle ;
\draw  [fill={rgb, 255:red, 255; green, 210; blue, 132 }  ,fill opacity=1 ][dash pattern={on 0.84pt off 2.51pt}] (67,90) -- (170,90) -- (170,113) -- (67,113) -- cycle ;
\draw  [fill={rgb, 255:red, 255; green, 210; blue, 132 }  ,fill opacity=1 ][dash pattern={on 0.84pt off 2.51pt}] (179,90) -- (282,90) -- (282,113) -- (179,113) -- cycle ;
\draw  [fill={rgb, 255:red, 80; green, 227; blue, 194 }  ,fill opacity=0.32 ][dash pattern={on 0.84pt off 2.51pt}] (466,90) -- (574,90) -- (574,113) -- (466,113) -- cycle ;
\draw  [fill={rgb, 255:red, 80; green, 227; blue, 194 }  ,fill opacity=0.32 ][dash pattern={on 0.84pt off 2.51pt}] (352,89) -- (460,89) -- (460,112) -- (352,112) -- cycle ;
\draw  [dash pattern={on 4.5pt off 4.5pt}]  (217.21,105.01) .. controls (245.14,168.86) and (389.91,164.21) .. (415.9,104.75) ;
\draw [shift={(417,102)}, rotate = 470.06] [fill={rgb, 255:red, 0; green, 0; blue, 0 }  ][line width=0.08]  [draw opacity=0] (10.72,-5.15) -- (0,0) -- (10.72,5.15) -- (7.12,0) -- cycle    ;
\draw [shift={(216,102)}, rotate = 69.81] [fill={rgb, 255:red, 0; green, 0; blue, 0 }  ][line width=0.08]  [draw opacity=0] (10.72,-5.15) -- (0,0) -- (10.72,5.15) -- (7.12,0) -- cycle    ;
\draw  [dash pattern={on 4.5pt off 4.5pt}]  (89.32,106.57) .. controls (103.55,223.1) and (475.66,256.11) .. (519.37,103.81) ;
\draw [shift={(520,101.5)}, rotate = 464.68] [fill={rgb, 255:red, 0; green, 0; blue, 0 }  ][line width=0.08]  [draw opacity=0] (10.72,-5.15) -- (0,0) -- (10.72,5.15) -- (7.12,0) -- cycle    ;
\draw [shift={(89,103)}, rotate = 86.66] [fill={rgb, 255:red, 0; green, 0; blue, 0 }  ][line width=0.08]  [draw opacity=0] (10.72,-5.15) -- (0,0) -- (10.72,5.15) -- (7.12,0) -- cycle    ;
\draw    (401.43,81.65) .. controls (441.14,44.77) and (506.72,49.58) .. (542.39,84.39) ;
\draw [shift={(544,86)}, rotate = 225.78] [fill={rgb, 255:red, 0; green, 0; blue, 0 }  ][line width=0.08]  [draw opacity=0] (10.72,-5.15) -- (0,0) -- (10.72,5.15) -- (7.12,0) -- cycle    ;
\draw [shift={(399,84)}, rotate = 315] [fill={rgb, 255:red, 0; green, 0; blue, 0 }  ][line width=0.08]  [draw opacity=0] (10.72,-5.15) -- (0,0) -- (10.72,5.15) -- (7.12,0) -- cycle    ;
\draw    (111.43,78.65) .. controls (151.14,41.77) and (216.72,46.58) .. (252.39,81.39) ;
\draw [shift={(254,83)}, rotate = 225.78] [fill={rgb, 255:red, 0; green, 0; blue, 0 }  ][line width=0.08]  [draw opacity=0] (10.72,-5.15) -- (0,0) -- (10.72,5.15) -- (7.12,0) -- cycle    ;
\draw [shift={(109,81)}, rotate = 315] [fill={rgb, 255:red, 0; green, 0; blue, 0 }  ][line width=0.08]  [draw opacity=0] (10.72,-5.15) -- (0,0) -- (10.72,5.15) -- (7.12,0) -- cycle    ;

\draw (174,58.4) node [anchor=north west][inner sep=0.75pt]    {$x$};
\draw (600,61.4) node [anchor=north west][inner sep=0.75pt]    {$y$};
\draw (64,60.4) node [anchor=north west][inner sep=0.75pt]    {$x-\alpha $};
\draw (312,60.4) node [anchor=north west][inner sep=0.75pt]    {$c$};

\end{tikzpicture}\vspace{-7mm}

	\caption{Illustration for the proof of Lemma \ref{lemma:2}. Dashed arrows show that the two strings are the reverse of each other. Solid arrows show that the two strings are equal.}
\end{figure}
   We denote the optimal solution by $A[\ell, \ell+d-1]$. Let $c = \ell + (d-1)/2$ be the center of the solution. Assume for the sake of contradiction that $c < (x+y)/2-2\alpha$. This means that $\ell < x-\alpha$. Thus, the interval $A[x-\alpha,x-1]$ is exactly equal to the interval $A[x,x+\alpha-1]$ and thus $A[x-\alpha,y]$ is also $\alpha$-periodic. This is in contradiction with the maximality of $[x,y]$. A similar proof holds for the case that $c > (x+y)/2+2\alpha$.
\end{proof}

Lemma \ref{lemma:2} implies that we only need to consider all candidates in $C$ that are within the range $[(x+y)/2-2\alpha,(x+y)/2+2\alpha]$. Moreover, the distance between every pair of candidates in $C$ is $\alpha/2$ and therefore we only need to consider a constant number of elements. Therefore, we can find the solution in time $\tilde O (\sqrt{d})$. 

\providecommand{\proc}{\mathrm{XLI_2}}
\providecommand{\Xloss}{\mathrm{Xloss}}
\providecommand{\Est}[1]{\textsf{QTest}(#1)}
\providecommand{\myLARGE}{\textsf{LARGE}}
\providecommand{\mySMALL}{\textsf{SMALL}}
\providecommand{\Error}{\textsf{ERROR}}
\section{$(1+\epsilon)$-Approximation of the Ulam distance}\label{sec:Ulam}
In this section we prove the following theorem.
\begin{theorem-repeat}{theorem:ulam}
\TheoremUlam
\end{theorem-repeat}

\subsection{Classical indicator for the Ulam distance}
Naumovitz, Saks and Seshadhri \cite{Naumovitz+SODA17} showed how to construct, for any constant $\epsilon>0$,  a classical algorithm that computes a $(1+\epsilon)$-approximation of the Ulam distance of two non-repetitive strings $A,B$ of length~$n$ in time $\tilde O(n/\edit(A,B)+\sqrt{n})$. Their algorithm is complex: it consists of nine procedures that form a hierarchy of gap tests and estimators, each with successively better run time. The core technique, which lies at the lowest layer of the hierarchy and gives a very good --- but slow --- estimation of the Ulam distance, is a variant of the Saks-Seshadhri algorithm for estimating the longest increasing sequence from~\cite{Saks+FOCS10}. 

For our purpose, we will only need the following result from~\cite{Naumovitz+SODA17}.\footnote{The precise statement of this result appears in Lemma 8.5 and Table 1 in~\cite{Naumovitz+SODA17}. The procedure is denoted $\textsf{XLI}_2$ in~\cite{Naumovitz+SODA17}. Note that the original statement is actually more general: it considers strings $A,B$ of different lengths and gives an indicator for a slightly different quantity (called $\Xloss(A,B)$ in \cite{Naumovitz+SODA17}). When $A$ and $B$ have the same length we have $\Xloss(A,B)=\frac{1}{2}\edit(A,B)$, which gives the statement of Proposition \ref{prop:Ulam}.} 
\begin{proposition}[\cite{Naumovitz+SODA17}]\label{prop:Ulam}
Let $\delta>0$ be any constant. For any two non-repetitive strings $A,B$ of length~$n$ and any integer parameter $t'\ge c\cdot\edit(A,B)$ for some constant $c$ depending only on $\delta$, there exists a $\tilde O(\sqrt{t'})$-time classical algorithm $\Estimator{A,B,\delta,t'}$ that outputs $1$ with some probability $p$ and $0$ with probability $(1-p)$, for some probability $p$ such that 
\[
p\in \left[(1-\delta)\cdot \frac{\edit(A,B)}{n+t'},(1+\delta)\cdot \frac{\edit(A,B)}{n+t'}+\frac{\delta}{n+t'}\right].
\]
\end{proposition}
The procedure of Proposition \ref{prop:Ulam} lies at one of the lowest layers of the hierarchy in \cite{Naumovitz+SODA17}, and works by applying the Saks-Seshadhri algorithm from~\cite{Saks+FOCS10} on randomly chosen substrings of $A$ and~$B$ (of size roughly $t$).\footnote{We stress that this procedure is still too slow to get a $\tilde O(n/\edit(A,B)+\sqrt{\edit(A,B)})$-time classical algorithm. Indeed, getting a good approximation of $\edit(A,B)$ based only on this procedure would require repeating it $\Omega((n+t)/\edit(A,B))$ times, which would require $\Omega(\sqrt{\edit(A,B)}\cdot (n+t)/\edit(A,B))=\Omega(n/\sqrt{\edit(A,B)})$ time. This is why several additional techniques (which lead to several additional layers in the hierarchy) are used in~\cite{Naumovitz+SODA17}.} 


\subsection{Quantum algorithm for the Ulam distance}
To prove Theorem \ref{theorem:ulam}, the basic idea is to apply quantum amplitude estimation on the classical algorithm \textsf{UlamIndic} from Proposition~\ref{prop:Ulam}. Let us explain this strategy more precisely and show the technical difficulties we need to overcome. Assume that we know that $\edit(A,B)$ is in the interval $[D_1,D_2]$ for some values $D_1$ and $D_2$. We can then use  $\Estimator{A,B,\delta,t}$ with $t=c\cdot D_2$, and apply the quantum amplitude estimation algorithm of Theorem~\ref{th:ae} to estimate the probability it outputs~1. If we use~$k=\Theta(\sqrt{(n+t)/D_1})=\Theta(\sqrt{n/D_1})$  in Theorem~\ref{th:ae} and $\delta$ small enough, we will get a good approximation of the quantity $\edit(A,B)/(n+t)$, and thus a good approximation of $\edit(A,B)$. The complexity of this strategy is 
$
\tilde O(\sqrt{n/D_1}\cdot \sqrt{D_2}),
$
which is $\tilde O(\sqrt{n})$  when $D_1\approx D_2$. 

The main issue is that we do not know such tight upper and lower bounds on $\edit(A,B)$. Concerning the upper bound, we overcome this difficultly by simply successively trying $D_2=n$, $D_2=(1-\eta) n$, $D_2=(1-\eta)^2 n$, $\ldots$ for some small constant $\eta$ (for technical reasons we actually start from $D_2=(1-\eta)(1-\epsilon)\sqrt{n}/c$, and deal with the case of larger $D_2$ using the classical algorithm from \cite{Naumovitz+SODA17}). For the lower bound, on the other hand, we cannot simply start with $D_1=1$ and iteratively increase this value, since the cost would be too high: in order to achieve an overall running time of $\tilde O(\sqrt{n})$ we need to keep $D_1\approx D_2$. Instead of estimating the probability that \textsf{UlamIndic} outputs 1 using Theorem \ref{th:ae}, which is too costly, we thus design a gap test (see Proposition \ref{prop:qtest} below) that enables us to check if this probability is larger than $D_2$ or smaller than $(1-\eta)D_2$ much more efficiently, in time $\tilde O(\sqrt{n/D_2})$. 

We now present the details of our quantum algorithm. Let us first present the gap test.
This test relies on quantum amplitude amplification. Here is the precise statement of quantum amplitude amplification that we will use.

\begin{theorem}[Theorem 12 in \cite{brassard2002quantum}]\label{th:ae}
Let $\Aa$ be a classical algorithm that runs in time $T$, outputs~$1$ with probability $p$ and outputs $0$ with probability $1-p$, for some (unknown) value $p\in [0,1]$. For any integer $k\ge 1$, there exists a quantum algorithm that runs in time $\tilde O(kT)$ and outputs with probability at least $8/\pi^2$ an estimate $\tilde p$ such that
\[
|p-\tilde p|\le 2\pi\frac{\sqrt{p(1-p)}}{k}+\frac{\pi^2}{k^2}.
\]
\end{theorem}

The gap test is described in the following proposition.
\begin{proposition}\label{prop:qtest}
Let $\Aa$ be a classical algorithm that runs in time $T$, outputs~$1$ with probability $p$ and outputs $0$ with probability $1-p$, for some (unknown) value $p$.
For any $q\in(0,1]$ and any constant $\eta\in(0,1]$, there exists a quantum algorithm denoted $\Est{\Aa,q,\eta}$ that runs in time $\tilde O(T/\sqrt{q})$ and with probability at least $1-\poly(n)$ outputs \myLARGE{} if $p\ge q$ and \mySMALL{} if $p\le(1-\eta)q$.
\end{proposition}
\begin{proof}
Figure \ref{fig:qtest} describes our main quantum gap test. The complexity of this test is $\tilde O(T/\sqrt{q})$, from Theorem \ref{th:ae}. We show below that its success probability is at least $8/\pi^2$. The success probability can then be increased to $1-1/\poly(n)$ by repeating the test $\Theta(\log n)$ times and using majority voting.
\begin{figure}[h!]
\begin{center}
\fbox{
\begin{minipage}{14.5 cm} 
\begin{itemize}
\item[1.]
Apply the algorithm of Theorem \ref{th:ae} with 
$
k=\frac{20}{\eta\sqrt{q}}.
$
Let $\tilde p$ denote the output.
\item[2.]
Output \myLARGE{} if $\tilde p\ge (1-\eta/2)q$ and output \mySMALL{} if $\tilde p<(1-\eta/2)q$.
\end{itemize}
\end{minipage}
}
\end{center}
\vspace{-4mm}
\caption{Quantum gap test.}\label{fig:qtest}
\end{figure}

Let us analyze the success probability of the quantum gap test of Figure \ref{fig:qtest}. Note that Theorem~\ref{th:ae} guarantees that with probability at least $8/\pi^2$ the value $\tilde p$ satisfies the following two inequalities:
\begin{align}\label{ineq1}
\tilde p&\le p+2\pi\eta\frac{\sqrt{p(1-p)}\sqrt{q}}{20} + \pi^2\eta^2 \frac{q}{400},
\end{align}
\begin{align}\label{ineq2}
\tilde p&\ge p-2\pi\eta\frac{\sqrt{p(1-p)}\sqrt{q}}{20} - \pi^2\eta^2 \frac{q}{400}.
\end{align}

Let us first consider the case $p\le (1-\eta)q$. In that case Inequality (\ref{ineq1}) implies
\[
\tilde p \le (1-\eta)q+2\pi\eta \frac{q}{20}+\pi^2\eta^2\frac{q}{400}< (1-\eta/2)q
\]
since $2\pi/20+\pi^2/400<1/2$.
The probability that the algorithm outputs \myLARGE{} is thus at most $1-8/\pi^2$.

Now consider the case $p\ge q$. In that case Inequality (\ref{ineq2}) implies
\begin{align*}
\tilde p&\ge p-2\pi\eta\frac{\sqrt{p}\sqrt{q}}{20} - \pi^2\eta^2 \frac{q}{400}\\
&= \sqrt{p}\left(\sqrt{p}-2\pi\eta\frac{\sqrt{q}}{20}\right)- \pi^2\eta^2 \frac{q}{400}\\
&\ge q -2\pi \eta\frac{q}{20}-\pi^2\eta^2\frac{q}{400}> (1-\eta/2)q.
\end{align*}
The probability that the algorithm outputs \mySMALL{} is thus at most $1-8/\pi^2$.
\end{proof}



We are now ready to give the proof of Theorem \ref{theorem:ulam}.

\begin{proof}[of Theorem \ref{theorem:ulam}]
If $\edit(A,B)\ge \frac{1-\epsilon}{c}\sqrt{n}$, where $c$ is the constant from Proposition \ref{prop:Ulam}, then $\edit(A,B)$ can already be computed in $\tilde O(\sqrt{n})$ time by the classical algorithm from \cite{Naumovitz+SODA17}. We thus assume below that $\edit(A,B)\le \frac{1-\epsilon}{c}\sqrt{n}$. We also assume that $\edit(A,B)>0$, since otherwise the two strings are identical, which can be checked in $\tilde O(\sqrt{n})$ time using Grover search.

Our quantum algorithm is described in Figure \ref{fig:qalg}. The algorithm invokes $\Estimator{A,B,\delta,\frac{c}{1-\epsilon}t}$ with several values of $t$ such that $t\le \frac{1-\epsilon}{c}\sqrt{n}$. The analysis of the correctness done below will show that with probability at least $1-1/\poly(n)$, all calls to \textsf{UlamIndic} performed by the algorithm satisfy the condition $t'\ge c\cdot\edit(A,B)$ required in the statement of Proposition \ref{prop:qtest} (in our case $t'=\frac{c}{1-\epsilon}t$). Now assume that the condition is satisfied and write $p_t$ the probability that $\Estimator{A,B,\delta,\frac{c}{1-\epsilon}t}$ outputs $1$. Observe that Proposition~\ref{prop:Ulam} guarantees that
\begin{equation}\label{rel}
p_t\in \left[(1-\delta)\left(1-\frac{1}{\sqrt{n}}\right)\cdot \frac{\edit(A,B)}{n},(1+\delta)\cdot \frac{\edit(A,B)}{n}+\frac{\delta}{n}\right]
\end{equation}
since we have $\frac{c}{1-\epsilon}t\le \sqrt{n}$.

\begin{figure}[h!]
\begin{center}
\fbox{
\begin{minipage}{13 cm} 
\begin{itemize}
\item[1.]
Set $\eta=\epsilon/3$ and $\delta=\epsilon/3$ and $r=\ceil{\frac{\log((1-\delta)(\sqrt{n}-1)/n)}{\log (1-\eta)}}$. 
\item[2.] 
For i from $1$ to $r$ do:
\begin{itemize}
\item[2.1.]
Set $t=(1-\eta)^i  \cdot\frac{1-\epsilon}{c}\sqrt{n}$ and $q=t/n$.
\item[2.2.]
Apply Algorithm $\Est{\Aa,q,\eta}$ with $\Aa=\Estimator{A,B,\delta,\frac{c}{1-\epsilon}t}$. 
\item[2.3.]
If the output is \myLARGE{}, then stop and return $t$.
\end{itemize}
\item[3.]
Return \Error{}.
\end{itemize}
\end{minipage}
}
\end{center}
\vspace{-2mm}
\caption{Quantum algorithm computing a $(1+\epsilon)$-approximation of the Ulam distance of two strings $A$ and $B$.}\label{fig:qalg}
\end{figure}

Then observe that for $r=\ceil{\frac{\log((1-\delta)(\sqrt{n}-1)/n)}{\log (1-\eta)}}$, the inequality 
\[
(1-\delta)\left(1-\frac{1}{\sqrt{n}}\right)\cdot \frac{\edit(A,B)}{n}\ge (1-\eta)^r/\sqrt{n}
\]
holds, since we are assuming that $\edit(A,B)\ge 1$.
From Proposition \ref{prop:qtest} combined with (\ref{rel}), this means that with probability at least $1-1/\poly(n)$, the algorithm will never reach Step $3$ (and thus does not output \Error{}). In the remaining of the proof, we assume that the algorithm stops before reaching Step 3. Let~$i^\ast$ denote the value of $i$ during the last iteration of the loop of Step 2, and write $t^\ast=(1-\eta)^{i^\ast}  \cdot\frac{1-\epsilon}{c}\sqrt{n}$ and $q^\ast=t^\ast/n$. Observe that the output of the algorithm is $t^\ast$. 

Assume first that $i^\ast=1$. The output is thus $t^\ast=(1-\eta)\frac{1-\epsilon}{c}\sqrt{n}$. Observe that Algorithm $\textsf{QTest}$ then necessarily outputted $\myLARGE$ for $i=1$. Proposition \ref{prop:qtest} guarantees that with probability at least $1-1/\poly(n)$ the inequality $p_{t^\ast}\ge (1-\eta)q^\ast$ holds (otherwise it would have outputted $\mySMALL$). 
This inequality combined with (\ref{rel}) implies that
\[
(1+\delta)\cdot \frac{\edit(A,B)}{n}+\frac{\delta}{n}\ge (1-\eta)q^\ast
\]
and thus the output $t^\ast=nq^\ast$ of the algorithm satisfies the following inequalities:
\begin{align}\label{ineq:UB}
\begin{aligned}
t^\ast&\le ((1+\delta)\cdot \edit(A,B)+\delta)/(1-\eta)\\
&\le (1+2\delta)\cdot\edit(A,B)/ (1-\eta)\\
&\le (1+2\delta)(1+2\eta) \cdot \edit(A,B)\\
&\le (1+\epsilon) \cdot \edit(A,B),
\end{aligned}
\end{align}
where we used the inequality $1/(1-x)\le 1+2x$, which holds for any $x\in[0,1/2]$. Also note that the inequality 
\[
t^\ast = (1-\eta)\cdot\frac{1-\epsilon}{c}\sqrt{n} \:\:\ge\:\: (1-\epsilon)\cdot\edit(A,B)
\]
is trivially satisfied since we are assuming $\edit(A,B)\le  \frac{1-\epsilon}{c}\sqrt{n}$ and $\eta<\epsilon$.
The output is thus a $(1+\epsilon)$-approximation of $\edit(A,B)$.

Assume now that $i^\ast\ge 2$. 
Observe that Algorithm $\textsf{QTest}$ then necessarily outputted $\myLARGE$ for $i=i^\ast$ and $\mySMALL$ for $i= i^\ast-1$. 
Since it outputted $\myLARGE$ for $i=i^\ast$, Inequalities (\ref{ineq:UB}) hold with probability at least $1-1/\poly(n)$, from exactly the same argument as above. Since the output was $\mySMALL$ for $i=i^\ast-1$, Proposition \ref{prop:qtest} guarantees that with probability at least $1-1/\poly(n)$ the inequality
$
p_{t^\ast/(1-\eta)}\le q^\ast/(1-\eta)
$
holds (otherwise it would have outputted $\myLARGE$).
Since $t^\ast=nq^\ast$, this inequality combined with (\ref{rel}) implies:
\begin{align*}
t^\ast&\ge n(1-\eta)p_{t^\ast/(1-\eta)}\\
&\ge (1-\eta)(1-\delta)\left(1-\frac{1}{\sqrt{n}}\right)\cdot \edit(A,B)\\
&= (1-\epsilon/3)(1-\epsilon/3)\left(1-\frac{1}{\sqrt{n}}\right)\cdot \edit(A,B)\\
&\ge (1-\epsilon)\cdot \edit(A,B).
\end{align*}
The output of the algorithm is thus a $(1+\epsilon)$-approximation of $\edit(A,B)$ as well. 

Since the complexity of applying Algorithm \textsf{UlamIndic} at Step $i$ is $\tilde O(\sqrt{(1-\eta)^i\sqrt{n}})$, the
overall complexity of the algorithm is 
\[
\tilde O\left(\sum_{i=1}^{r} \frac{1}{\sqrt{(1-\eta)^i/\sqrt{n}}}\sqrt{(1-\eta)^i\sqrt{n}}\right)
=
\tilde O\left(\sqrt{n}\right),
\]
as claimed.
\end{proof}
\providecommand{\proc}{\mathrm{XLI_2}}

\section{Lower bounds}\label{sec:lb}
In this section we prove lower bounds for the problems considered in this paper.

We start by an easy lower bound for the longest common substring over a large alphabet.
\begin{theorem}\label{theorem:lcs-approx-LB-easy}
For any constant $c\in(0,1]$, any quantum algorithm that computes with high probability a $c$-approximation
of the longest common substring of two strings of length $n$ over an alphabet of size $2n$ requires $\Omega(n^{2/3})$ time. This lower bound also holds for non-repetitive strings.
\end{theorem}
\begin{proof}
Consider the following version of the element distinctness problem: given a list $L$ of $m$ characters in an alphabet of size $m$ such that either all the characters of $L$ are distinct, or only one character occurs twice in $L$ (i.e., the other $m-2$ characters are distinct), decide which of the two cases holds. A $\Omega(m^{2/3})$-query lower bound is known for this problem in the quantum setting \cite{aaronson04,ambainis05,kutin05}.

Let us take $m=2n$. Construct a string $A$ of length $n$ by taking $n$ elements from $L$ uniformly at random, and construct another string $B$ of length $n$ using the remaining $n$ elements from $L$. Observe that if all the characters in $L$ are distinct then $A$ and $B$ are non-repetitive and have no common substring. On the other hand, if the characters in $L$ are not all distinct, then with probability at least $1/2$ the two strings $A$ and $B$ are non-repetitive and have a common substring of length $1$. This gives a (randomized) reduction from element distinctness problem to our problem since a $c$-approximation of the longest common substring enables us to distinguish between the two cases.
\end{proof}

We now present a lower bound for \textsf{LCS} that holds even for binary strings.
\begin{theorem}\label{theorem:lcs-approx-LB}
For any constant $c\in(0,1]$, any quantum algorithm that computes with high probability a $c$-approximation of longest common substring of two binary strings of length $n$ requires $\tilde  \Omega(n^{2/3})$ time.
\end{theorem}
\begin{proof}
The main idea is simple: by dividing the input strings into blocks of size $\Theta(\log n)$, we can reduce the case with alphabet of size $2n$ to the case of binary alphabets, and thus use the lower bound of Theorem \ref{theorem:lcs-approx-LB-easy} with only a logarithmic overhead.

We now give more details of the reduction. Let $A$ and $B$ denote strings of length $n$ over an alphabet of size $2n$. Each character of the alphabet is encoded by a random binary string of length $s$ (i.e., a binary string of length $s$ such that each bit is 1 with probability $1/2$). Using easy arguments from probability theory (see, e.g., Section 2 of \cite{Arratia+85}), for any constant $\alpha>0$ we can guarantee that the following property holds with probability at least $9/10$ if we take $s\ge d_\alpha \cdot \log n$ for some constant~$d_\alpha$ that depends only on $\alpha$: the length of the longest common substring of the encodings of any two distinct characters is at most $\alpha s$. Let us choose $\alpha=c/3$ and consider $s=\ceil{d_\alpha \cdot \log n}$. Below we assume that the above property holds (with happens with probability at least 9/10).

If the longest common substring of $A$ and $B$ has length zero then the longest common substring of their binary versions has length at most $2\alpha s=2cs/3< cs$. If the longest common substring of $A$ and $B$ has length at least one, on the other hand, then the longest common substring of their binary versions has length at least $s$. Thus a $c$-approximation of the longest common substring of the binary versions enables us to distinguish, with high probability, between the two cases.
\end{proof}

We now give our lower bounds for \textsf{LPS} and for the computation of the Ulam distance. 
\begin{theorem}\label{theorem:lps-LB}
For any constant $c\in (0,1]$, any quantum algorithm that computes with high probability a $c$-approximation of the longest palindrome of a binary string of length $n$ requires $\Omega(\sqrt{n})$ time.
\end{theorem}
\begin{proof}
Let $m\ge 3$ be an integer. Let $S_1\subseteq\{0,1\}^m$ denote the set of all $m$-bit strings of Hamming weight one in which the first and last characters are both zero, and let $S_0\subseteq \{0,1\}^m$ be the set containing only the all-zero string. Distinguishing between strings in $S_0$ and $S_1$ requires 
 $\Omega(\sqrt{m})$ queries in the quantum setting (see, e.g., \cite{bennett97}).

Let us write $k=\ceil{3/c}$.
Given a string $x\in S_0\cup S_1$ and any $r\in\{1,\ldots,k\}$, let $x^r$ be the binary string of length $km$ obtained from $x$ by  replacing each $0$ in $x$ by $0^k$ and each $1$ in $x$ by $1^r0^{k-r}$. 

Take $n=k^2m$. We now consider the string $A\in\{0,1\}^{n}$ obtained by concatenating the strings  $x^1$, $x^2$, $\ldots$, $x^{k}$. Each $x^i$, for $i\in\{1,\ldots,k\}$, is called a block of $A$. Observe that if $x\in S_0$ (i.e., if $x$ is the all-zero string), then $A$ is also the all-zero strings. In this case the length of the \textsf{LPS} of $A$ is $n$. On the other hand, if $x\in S_1$ then no palindrome of $A$ can include two full blocks (since the numbers of repeated 1s in distinct blocks do not match), and thus the length of the \textsf{LPS} of $A$ is at most $km+2(km-1)\le 3km< cn$. Thus computing a $c$-approximation of the \textsf{LPS} of $A$ enables us to distinguishing between strings in $S_0$ and~$S_1$. This gives the lower bound $\Omega(\sqrt{m})=\Omega(\sqrt{n})$ on the complexity of computing a $c$-approximation of the \textsf{LPS}  .
\end{proof}
\begin{theorem}\label{theorem:ulam-approx-LB}
For any $\epsilon\in[0,1)$, any quantum algorithm that computes with high probability a $(1+\epsilon)$-approximation of the Ulam distance of two non-repetitive strings requires $ \Omega(\sqrt{n})$ time.
\end{theorem}
\begin{proof}
Let us consider the alphabet $\{1,2,\ldots,n\}$. Let $A$ be the all-increasing string, i.e., $A=123\cdots n$. Let $B$ either the all-increasing string or the string obtained by permuting the $\ell$-th position and the $(\ell+1)$-position of the all-increasing string for some unknown $\ell\in\{1,\ldots,n-1\}$. Note that in former case the Ulam distance of $A$ and $B$ is zero, while in the second case the Ulam distance is two. Computing a $(1+\epsilon)$-factor approximation of the Ulam distance of $A$ and $B$ thus requires distinguishing between the two cases, which requires $\Omega(\sqrt{n})$ queries from the lower bound on Grover search \cite{bennett97}.
\end{proof}
%

\section*{Acknowledgments}
The authors are grateful to Michael Saks and C. Seshadhri for helpful correspondence. FLG was supported by the JSPS KAKENHI grants Nos.~JP16H01705, JP19H04066, JP20H00579 and JP20H04139, and by the MEXT Quantum Leap Flagship Program (MEXT Q-LEAP) grant No. JPMXS0120319794.
\bibliographystyle{alpha}	
\bibliography{draft}

\end{document}